\documentclass[a4paper,UKenglish,cleveref, autoref, thm-restate]{lipics-v2021}

\hideLIPIcs  
\nolinenumbers

\title{Going Deep and Going Wide: Counting Logic and Homomorphism Indistinguishability over Graphs of Bounded Treedepth and Treewidth} 
\titlerunning{Going Deep and Going Wide} 
\author{Eva Fluck}{RWTH Aachen University, Germany}{fluck@cs.rwth-aachen.de}{https://orcid.org/0000-0002-9643-6081}{}

\author{Tim Seppelt}{RWTH Aachen University, Germany}{seppelt@cs.rwth-aachen.de}{https://orcid.org/0000-0002-6447-0568}{German Research Council (DFG) via Research Training Group 2236 (UnRAVeL)}

\author{Gian Luca Spitzer}{RWTH Aachen University, Germany}{gian.luca.spitzer@rwth-aachen.de}{https://orcid.org/0009-0008-0270-506X}{German Research Council (DFG) via Research Training Group 2236 (UnRAVeL)}

\authorrunning{E. Fluck, T. Seppelt, and G. L. Spitzer} 
\Copyright{Eva Fluck, Tim Seppelt, and Gian Luca Spitzer} 


\category{} 
\relatedversion{}

\acknowledgements{We would like to thank Martin Grohe and Daniel Neuen for fruitful discussions.}

\renewcommand{\phi}{\varphi}
\renewcommand{\epsilon}{\varepsilon}
\renewcommand{\theta}{\vartheta}
\usepackage{amsmath}
\usepackage{amssymb}
\usepackage{mathtools}

\usepackage{subcaption}
\usepackage{tikz}
\usetikzlibrary{math}
\usetikzlibrary{positioning,arrows.meta,fit,snakes,shapes}

\definecolor{rwth-blue}{cmyk}{1,.5,0,0}\colorlet{rwth-lblue}{rwth-blue!50}\colorlet{rwth-llblue}{rwth-blue!25}
\definecolor{rwth-violet}{cmyk}{.6,.6,0,0}\colorlet{rwth-lviolet}{rwth-violet!50}\colorlet{rwth-llviolet}{rwth-violet!25}
\definecolor{rwth-purple}{cmyk}{.7,1,.35,.15}\colorlet{rwth-lpurple}{rwth-purple!50}\colorlet{rwth-llpurple}{rwth-purple!25}
\definecolor{rwth-carmine}{cmyk}{.25,1,.7,.2}\colorlet{rwth-lcarmine}{rwth-carmine!50}\colorlet{rwth-llcarmine}{rwth-carmine!25}
\definecolor{rwth-red}{cmyk}{.15,1,1,0}\colorlet{rwth-lred}{rwth-red!50}\colorlet{rwth-llred}{rwth-red!25}
\definecolor{rwth-magenta}{cmyk}{0,1,.25,0}\colorlet{rwth-lmagenta}{rwth-magenta!50}\colorlet{rwth-llmagenta}{rwth-magenta!25}

\definecolor{rwth-yellow}{cmyk}{0,0,1,0}\colorlet{rwth-lyellow}{rwth-yellow!50}\colorlet{rwth-llyellow}{rwth-yellow!25}
\definecolor{rwth-grass}{cmyk}{.35,0,1,0}\colorlet{rwth-lgrass}{rwth-grass!50}\colorlet{rwth-llgrass}{rwth-grass!25}
\definecolor{rwth-green}{cmyk}{.7,0,1,0}\colorlet{rwth-lgreen}{rwth-green!50}\colorlet{rwth-llgreen}{rwth-green!25}
\definecolor{rwth-cyan}{cmyk}{1,0,.4,0}\colorlet{rwth-lcyan}{rwth-cyan!50}\colorlet{rwth-llcyan}{rwth-cyan!25}
\definecolor{rwth-teal}{cmyk}{1,.3,.5,.3}\colorlet{rwth-lteal}{rwth-teal!50}\colorlet{rwth-llteal}{rwth-teal!25}
\definecolor{rwth-gold}{cmyk}{.35,.46,.7,.35}
\definecolor{rwth-silver}{cmyk}{.39,.31,.32,.14}

\newcommand{\R}{\ensuremath{\mathbb{R}}}
\newcommand{\N}{\ensuremath{\mathbb{N}}}

\newcommand{\I}{\ensuremath{\mathfrak I}}

\newcommand{\TW}{\ensuremath{\mathcal{TW}}}
\newcommand{\TD}{\ensuremath{\mathcal{TD}}}

\newcommand{\Ekq}{\ensuremath{\mathcal{T}_{q}^{k}}}
\newcommand{\EParam}[2]{\ensuremath{\mathcal{T}_{#2}^{#1}}}
\newcommand{\Lkq}{\ensuremath{\mathcal{L}_{q}^{k}}}
\newcommand{\LParam}[2]{\ensuremath{\mathcal{L}_{#2}^{#1}}}
\newcommand{\GEkq}{\ensuremath{\mathcal{G}\Lkq}}
\newcommand{\GEParam}[2]{\ensuremath{\mathcal{G}\LParam{#1}{#2}}}

\newcommand{\GEkqLL}{\ensuremath{\mathcal{G}\Ekq}}

\newcommand{\CG}{\ensuremath{\mathcal G}}

\newcommand{\lFO}{\ensuremath{\mathsf{FO}}}
\newcommand{\lC} {\ensuremath{\mathsf C}}
\newcommand{\lL} {\ensuremath{\mathsf L}}
\newcommand{\lGC}{\ensuremath{\mathsf{GC}}}

\newcommand{\impl}{\ensuremath{\rightarrow}}

\newcommand{\parto}{\ensuremath{\rightharpoonup}}

\newcommand{\tup}[1]{\ensuremath{\boldsymbol{#1}}}
\newcommand{\qg}[1]{\ensuremath{\mathfrak{#1}}} \newcommand{\qgp}[3]{\ensuremath{\mathfrak{#1}[#3; #2]}}

\DeclareMathOperator{\CR}{CR}
\DeclareMathOperator{\monCR}{mon-CR}

\DeclareMathOperator{\qr}{qr}

\DeclareMathOperator{\free}{free}

\DeclareMathOperator{\dom}{dom}
\DeclareMathOperator{\img}{img}
\DeclareMathOperator{\HOM}{Hom}

\DeclareMathOperator{\tw}{tw}
\DeclareMathOperator{\td}{td}

\DeclareMathOperator{\dep}{dp}

\DeclareMathOperator{\cl}{cl}
\newcommand\restrict[1]{\vert_{#1}}

\renewcommand{\vec}[1]{\ensuremath{\boldsymbol{#1}}}

\newcommand{\elimOrd}[2]{$(#1,#2)$-constructible}

\newcommand{\elimDepth}{elimination depth}
\newcommand{\labels}[1]{\ensuremath{L_{#1}}}
\newcommand{\grid}[2]{\ensuremath{G_{#1 \times #2}}}

\makeatletter
\let\amsmath@bigm\bigm

\renewcommand{\bigm}[1]{\ifcsname fenced@\string#1\endcsname
    \expandafter\@firstoftwo
  \else
    \expandafter\@secondoftwo
  \fi
  {\expandafter\amsmath@bigm\csname fenced@\string#1\endcsname}{\amsmath@bigm#1}}

\newcommand{\DeclareFence}[2]{\@namedef{fenced@\string#1}{#2}}
\makeatother

\DeclareFence{\mid}{|}

\newtheorem{fact}[definition]{Fact}
\Crefname{fact}{Fact}{Facts}

\begin{document}
\colorlet{rwth-orange}{lipicsYellow}
\maketitle

\begin{abstract}
  We study the expressive power of first-order logic with counting quantifiers, especially the $k$-variable and quantifier-rank-$q$ fragment $\mathsf{C}^k_q$, using homomorphism indistinguishability. 
Recently, Dawar, Jakl, and Reggio~(2021) proved that two graphs satisfy the same $\mathsf{C}^k_q$-sentences if and only if they are homomorphism indistinguishable over the class $\mathcal{T}^k_q$ of graphs admitting a $k$-pebble forest cover of depth $q$. Their proof builds on the categorical framework of game comonads developed by Abramsky, Dawar, and Wang~(2017). We reprove their result using elementary techniques inspired by Dvo\v{r}ák~(2010). Using these techniques we also give a characterisation of guarded counting logic.
Our main focus, however, is to provide a graph theoretic analysis of the graph class $\mathcal{T}^k_q$. This allows us to separate $\mathcal{T}^k_q$ from the intersection of the graph class $\mathcal{TW}_{k-1}$, that is graphs of treewidth less or equal $k-1$, and $\mathcal{TD}_q$, that is graphs of treedepth at most $q$ if $q$ is sufficiently larger than $k$. We are able to lift this separation to the semantic separation of the respective homomorphism indistinguishability relations. A part of this separation is to prove that the class $\mathcal{TD}_q$ is homomorphism distinguishing closed, which was already conjectured by Roberson~(2022).

 \end{abstract}

\section{Introduction}

Since the 1980s, first-order logic with counting quantifiers $\mathsf{C}$ plays a decisive role in finite model theory.
In this extension of first-order logic with quantifiers $\exists^{\geq t} x$ (``there exists at least $t$ many $x$''), properties which can be expressed in first-order logic only with formulae of length depending on $t$ can be expressed succinctly.
Of particular interest are the $k$-variable and quantifier-depth-$q$ fragments $\mathsf{C}^k$ and $\mathsf{C}_q$ of $\mathsf{C}$, which
enjoy rich connections to graph algorithms \cite{dvorak_recognizing_2010}, algebraic graph theory \cite{dell_lovasz_2018,grohe_homomorphism_2022}, optimisation \cite{grohe_homomorphism_2022,roberson_lasserre_2023}, graph neural networks \cite{grohe_logic_2021}, and category theory \cite{dawar_lovasz-type_2021,abramsky_pebbling_2017}.

The intersection of these fragments, the fragment $\mathsf{C}^k_q \coloneqq \mathsf{C}^k \cap \mathsf{C}_q$ of all $\mathsf{C}$-formulae with $k$-variables and quantifier-depth $q$, has received much less attention \cite{rattan_weisfeiler-leman_2023}. In this work, we study the expressivity of $\mathsf{C}^k_q$ using homomorphism indistinguishability.

Homomorphism indistinguishability is an emerging framework for measuring the expressivity of equivalence relations comparing graphs. Two graphs $G$ and $H$ are \emph{homomorphism indistinguishable} over a graph class $\mathcal{F}$ if for all $F \in \mathcal{F}$ the number of homomorphisms from $F$ to $G$ is equal to the number of homomorphisms from $F$ to $H$. Many natural equivalence relations between graphs including isomorphism \cite{lovasz_operations_1967}, quantum isomorphism \cite{mancinska_quantum_2020}, cospectrality \cite{dell_lovasz_2018}, and feasibility of integer programming relaxations for graph isomorphism \cite{grohe_homomorphism_2022,roberson_lasserre_2023} can be characterised as homomorphism indistinguishability relations over certain graph classes. Establishing such characterisations is intriguing since it allows to use tools from structural graph theory to study equivalence relations between graphs \cite{roberson_oddomorphisms_2022,seppelt_logical_2023}.
Furthermore, the expressivity of homomorphism counts themselves is of practical interest \cite{nguyen_graph_2020,grohe_word2vec_2020}.

Equivalence with respect to $\mathsf{C}^{k}$ and $\mathsf{C}_q$ has been characterised by Dvo\v{r}ák~\cite{dvorak_recognizing_2010} and Grohe~\cite{grohe_counting_2020} as homomorphism indistinguishability over the classes $\mathcal{TW}_{k-1}$ of graphs of treewidth $\leq k-1$ and $\mathcal{TD}_q$ of graphs of treedepth $\leq q$, respectively.
Recently, Dawar, Jakl, and Reggio \cite{dawar_lovasz-type_2021} proved that two graphs satisfy the same $\mathsf{C}^k_q$-sentences if and only if they are homomorphism indistinguishable over the class $\Ekq$ of graphs admitting a $k$-pebble forest cover of depth $q$. Their proof builds on the categorical framework of game comonads developed in~\cite{abramsky_pebbling_2017}.

As a first step, we reprove their result using elementary techniques inspired by Dvo\v{r}\'ak~\cite{dvorak_recognizing_2010}. The general idea is to translate between sentences in $\mathsf{C}$ and graphs from which homomorphism are counted in an inductive fashion. By carefully imposing structural constraints, we are able to extend the original correspondence from \cite{dvorak_recognizing_2010} between $\mathsf{C}^{k}$ and graphs of treewidth at most $k-1$ to $\mathsf{C}_q$ and graphs of treedepth at most $q$, reproducing a result of \cite{grohe_counting_2020}, and finally to $\mathsf{C}^k_q$ and $\Ekq$. This simple and uniform proof strategy also yields the following result on guarded counting logic $\mathsf{GC}^k_q$.
Guarded counting logic plays a crucial role in the theory of properties of higher arity expressible by graph neural networks \cite{grohe_logic_2021}.
Towards this goal we introduce a new graph class called \GEkqLL, which is closely related to \Ekq.

\begin{restatable}{theorem}{guardedEkqLogic}
	\label{thm:guardedEkq_vs_guarded-logic}
	Let $k, q \geq 1$.
	Two graphs $G$ and $H$ are $\lGC^k_q$-equivalent
	if and only if 
	they are homomorphism indistinguishable over $\GEkqLL$.
\end{restatable}

The main contribution of this work, however, concerns the relationship between the graph classes $\Ekq$ and the class $\mathcal{TW}_{k-1} \cap \mathcal{TD}_q$ of graphs which have treewidth at most $k-1$ \emph{and} treedepth at most $q$.
Given the results of \cite{dvorak_recognizing_2010,grohe_counting_2020}, one might think that elementary equivalence with respect to sentences in $\mathsf{C}^{k}_q = \mathsf{C}^{k} \cap \mathsf{C}_q$ is characterised by homomorphism indistinguishability with respect to $\mathcal{TW}_{k-1} \cap \mathcal{TD}_q$. The central result of this paper asserts that this intuition is wrong.
As a fist step towards this, we prove that the graph class $\Ekq$ and $ \TW_{k-1}\cap\TD_q$ are distinct if $q$ is sufficiently larger\footnote{All logarithms in this work are to the base $2$.} than $k$.

\begin{restatable}{theorem}{ekqtwtdsyntax}
	\label{thm:Ekq_tw-td}
	For  $q\geq 3$ and $2\leq k-1\leq\frac{q}{3+\log q}$, $\Ekq\subsetneq \TW_{k-1}\cap\TD_q$ and $\EParam{2}{q}\subsetneq \TW_{1}\cap\TD_q$.
\end{restatable}

Towards \cref{thm:Ekq_tw-td}, we give an equivalent characterisation of $\Ekq$ via a monotone cops-and-robber game, which is essentially the standard game for treewidth where one additionally counts the number of rounds the cops need to capture the robber.
Here, ``monotone'' refers to a restriction of Cops, who is only allowed to move to positions that do not enlarge the escape-space of Robber.
Building on \cite{Furer01}, we then prove that $\Ekq$ is a proper subclass of $\mathcal{TW}_{k-1} \cap \mathcal{TD}_q$, for $q$ sufficiently larger than $k$.
Additionally, we provide an analysis of various notions designed to restrict both width and depth of a decomposition and show that all of them are equivalent.
Adding to the original definition of $\Ekq$ via $k$-pebble forest covers of depth $q$, which can be interpreted as treedepth decompositions augmented by a width measure, we introduce a way to measure the depth of tree decompositions. Finally, we define $k$-construction trees of elimination depth $q$, another equivalent notion, which relates to the machinery used by Dvo\v{r}\'ak~\cite{dvorak_recognizing_2010}.

However, the, let us say syntactical, separation of the graph classes $\Ekq$ and $\mathcal{TW}_{k-1} \cap \mathcal{TD}_q$ from \cref{thm:Ekq_tw-td} does not suffice to separate their homomorphism indistinguishability relations semantically. In fact, it could well be that all graphs which are homomorphism indistinguishable over $\Ekq$ are also homomorphism indistinguishable over $\mathcal{TW}_{k-1} \cap \mathcal{TD}_q$.

That such phenomena do not arise under certain mild assumptions was recently conjectured by Roberson~\cite{roberson_oddomorphisms_2022}. His conjecture asserts that every graph class which is closed under taking minors and disjoint unions is homomorphism distinguishing closed. Here, a graph class $\mathcal{F}$ is \emph{homomorphism distinguishing closed} if it satisfies the following maximality condition: For every graph $F \not\in \mathcal{F}$, there exists two graphs $G$ and $H$ which are homomorphism indistinguishable over $\mathcal{F}$ but have different numbers of homomorphism from $F$.

Since $\Ekq$, $\mathcal{TW}_{k-1}$, and $\mathcal{TD}_q$ are closed under disjoint unions and minors, the confirmation of Roberson's conjecture would readily imply the semantic counterpart of \cref{thm:Ekq_tw-td}. Unfortunately, Roberson's conjecture is wide open and has been confirmed only for the class of all planar graphs \cite{roberson_oddomorphisms_2022}, $\mathcal{TW}_{k-1}$ \cite{neuen_homomorphism-distinguishing_2023}, and for graph classes which are essentially finite \cite{seppelt_logical_2023}. Guided by \cite{neuen_homomorphism-distinguishing_2023}, we add to this short list of examples:

\begin{restatable}{theorem}{tdclosed}
	\label{thm:td-closed}
	For $q \geq 1$, the class $\mathcal{TD}_q$ is homomorphism distinguishing closed.
\end{restatable}

Combining this with the results of \cite{neuen_homomorphism-distinguishing_2023}, we get that $\TW_{k-1}\cap\TD_q$ is homomorphism distinguishing closed as well.
We then set out to separate homomorphism indistinguishability over $\Ekq$ and $\TW_{k-1}\cap\TD_q$.
Despite not being able to prove that $\Ekq$ is homomorphism distinguishing closed, we prove that the homomorphism distinguishing closure of $\Ekq$, i.e.\@ the smallest homomorphism distinguishing closed superclass of $\Ekq$, is a proper subclass of $\TW_{k-1}\cap\TD_q$, for $q$ sufficiently larger than $k$.
Written out, \cref{thm:Ekq_tw-td-semantics} asserts the following:
Whenever $q$ is sufficiently large in terms of $k$,
(1) there exist graphs which are homomorphism indistinguishable over $\Ekq$ but not over $\TW_{k-1}\cap\TD_q$, and
(2) for every graph $F \not\in \mathcal{TW}_{k-1} \cap \mathcal{TD}_q$, there exist graphs $G$ and $H$ which are homomorphism indistinguishable over $\Ekq$ but have a different number of homomorphisms from $F$.

\begin{restatable}{theorem}{ekqtwtdsemantics}
	\label{thm:Ekq_tw-td-semantics}
	For $q\geq 1$ and $2 \leq k-1\leq\frac{q}{3+\log q}$,  $\cl(\Ekq) \subsetneq \TW_{k-1}\cap\TD_q$.
\end{restatable}

Besides obtaining \cref{thm:Ekq_tw-td-semantics}, we distil the challenge of proving that $\Ekq$ is homomorphism distinguishing closed to the question whether the monotone variant of the cops-and-robber game is equivalent to the non-monotone variant.
In general this equivalence between the monotone and non-monotone variant of a graph searching game is a non-trivial property.
There are games where the two variants are equivalent, such as the games corresponding to treewidth \cite{seymour_graph_1993} and treedepth \cite{giannopoulou_lifo-search_2012}, as well as games where they are not, such as games corresponding to directed treewidth \cite{KreutzerO11} or hypertreewidth \cite{Gottlob03}.

\section{Preliminaries}
\subparagraph{Notation.}
By $[k]$ we denote the set $\{1, \dots, k\}$. For a function $f$, we denote the domain of $f$ by $\dom(f)$. The \emph{image} of $f$ is the set $\img(f) \coloneqq \{f(x) ~|~ x \in \dom(f)\}$.
The \emph{restriction} of a function $f \colon A \to C$ to some set $B \subseteq A$ is the function $f\restrict{B} \colon B \to C$ with $f\restrict{B}(x) = f(x)$ for $x \in B$.
For functions $f \colon A \to C$, $g \colon B \to C$ that agree on $A \cap B$, we write $f \sqcup g$ for the union of $f$ and $g$, that is, the function mapping $x$ to $f(x)$ if $x \in A$ and to $g(x)$ if $x \in B$.

We use bold letters to denote tuples. The tuple elements are denoted by the corresponding regular letter together with an index. For example, $\tup{a}$ stands for the tuple $(a_1, \dots, a_n)$.

\subparagraph{Graphs and Labels.}
A graph $G$ is a tuple $(V(G), E(G))$, where $V(G)$ is a finite set of vertices and $E(G) \subseteq \binom{V(G)}{2}$ is the set of edges. We usually write $uv$ or $vu$ to denote the edge $\{u, v\} \in E(G)$. Unless otherwise specified, all graphs are assumed to be simple: They are undirected, unweighted and contain neither loops nor parallel edges. We denote the class of all graphs by $\CG$.

A \emph{$k$-labelled graph} $G$ is a graph together with a partial function $\nu_G \colon [k] \parto V$ that assigns labels from the finite set $[k] = \{1, \dots, k\}$ to vertices of $G$. A label thus occurs at most once in a graph, a single vertex can have multiple labels, and not all labels have to be assigned. A graph where every vertex has at least one label is called \emph{fully labelled}.
For $\ell \in [k]$ and $v \in V(G)$, we write $G(\ell \to v)$ to denote the graph obtained from $G$ by setting $\nu_{G(\ell \to v)}(\ell) = v$. We can \emph{remove} a label $\ell$ from a graph $G$, which yields a copy $G'$ of $G$ where $\nu_{G'}(\ell) = \bot$ and $\nu_{G'}(\ell') = \nu_{G}(\ell')$ for all $\ell' \neq \ell$. The \emph{product} $G_1G_2$ of two labelled graphs is the graph obtained by taking the disjoint union of $G_1$ and $G_2$, identifying vertices with the same label, and suppressing any loops or parallel edges that might be created.
We denote the class of all $k$-labelled graphs by $\CG_k$.

We call $H$ a \emph{subgraph} of $G$ if $H$ can be obtained from $G$ by removing vertices and edges. $H$ is a \emph{minor} of $G$ if it can be obtained from $G$ by removing vertices, removing edges, and \emph{contracting} edges. We contract an edge $uv$ by removing it and identifying its incident vertices. For labelled graphs, the new vertex is labelled by the union of labels of $u$ and $v$.

A graph is \emph{connected} if there exists a path between any two vertices. A \emph{tree} is a graph where any two vertices are connected by exactly one path. The disjoint union of one or more trees is called a \emph{forest}. A \emph{rooted tree} $(T, r)$ is a tree $T$ together with some designated vertex $r \in V(T)$, the \emph{root} of $T$. A \emph{rooted forest} $(F, \tup{r})$ is a disjoint union of rooted trees. The \emph{height} of a rooted tree is equal to the number of vertices on the longest path from the root to the leaves. The height of a rooted forest is the maximum height over all its connected components.

At times, the following alternative definition is more convenient. We can view a rooted forest $(F, \tup{r})$ as a pair $(V(F), \preceq)$, where $\preceq$ is a partial order on $V(F)$ and for every $v \in V(F)$ the elements of the set $\{u \in V(F) \mid u \preceq v\}$ are pairwise comparable: The minimal elements of $\preceq$ are precisely the roots of $F$, and for any rooted tree $(T, r)$ that is part of $F$ we let $v \preceq w$ if $v$ is on the unique path from $r$ to $w$.

The height of a rooted forest $(F, \tup{r})$ is then given by the length of the longest $\preceq$-chain. A rooted tree $(T', r')$ is a \emph{subtree} of a tree $(T, r)$ if $V(T') \subseteq V(T)$ and $\preceq^{T'}$ is the restriction of $\preceq^T$ to $V(T')$. Note that the subgraph of $T$ induced by $V(T')$ might not be a tree, since the vertices of $T'$ can be interleaved with vertices that do not belong to $T'$. We call a subtree $T'$ of $T$ \emph{connected} if its induced subgraph on $T$ is connected.

\subparagraph{Homomorphisms.}
A \emph{homomorphism} from a graph $F$ to a graph $G$ is a mapping $h \colon V(F) \to V(G)$ that satisfies $uv \in E(F) \implies h(u)h(v) \in E(G)$. For $k$-labelled graphs, we additionally require that $h(\nu_F(\ell)) = \nu_G(\ell)$ for all $\ell \in [k]$.
We denote the set of homomorphisms from $F$ to $G$ by $\HOM(F, G)$. The number of homomorphisms from $F$ to $G$ we denote by $\hom(F, G) \coloneqq |\HOM(F, G)|$. We also write $\HOM(F, G; a_1 \mapsto b_1, \dots, a_n \mapsto b_n)$ to denote the set of homomorphism $h\colon F \to G$ satisfying $h(a_i) = b_i$ for $i \in [n]$.
Two graphs $G$ and $H$ are \emph{homomorphism indistinguishable} over a class of graphs $\mathcal{F}$ if $\hom(F, G) = \hom(F, H)$ for all $F \in \mathcal{F}$.

\subparagraph{Logic of Graphs.}

We will mainly consider \emph{counting first-order logic} $\lC$. $\lC$ extends regular first-order logic $\lFO$ by quantifiers $\exists^{\geq t}$, for $t \in \N$. Consequently, we can build a $\lC$-formula in the usual way from atomic formulae; variables $x_1, x_2, \dots$; logical operators $\land, \lor, \impl, \neg$; and quantifiers $\forall, \exists, \exists^{\geq t}$. The atomic formulae in the language of graphs are $E\alpha\beta$ and $\alpha = \beta$ for arbitrary variables $\alpha, \beta$.

An occurrence of a variable $x$ is called \emph{free} if it is not in the scope of any quantifier. The \emph{free variables} $\free(\phi)$ of a formula $\phi$ are precisely those that have a free occurrence in $\phi$. A formula without free variables is called a \emph{sentence}. We often write $\phi(x_1, \dots, x_n)$ to denote that the free variables of $\phi$ are among $x_1, \dots, x_n$. For a graph $G$, it usually depends on the interpretation of the free variables whether $G \models \phi(x_1, \dots, x_n)$. We write $G, v_1, \dots, v_n \models \phi(x_1, \dots, x_n)$ or $G \models \phi(v_1, \dots, v_n)$ if $G$ satisfies $\phi$ when $x_i$ is interpreted by $v_i$. We might also give an explicit \emph{interpretation function} $\I \colon \free(\phi) \to V(G)$, writing $G, \I \models \phi$.

We generalise the notion of $\lL$-equivalence, writing $G, v_1, \dots, v_n \equiv_\lL H, w_1, \dots, w_n$ to denote that for all formulae $\phi(x_1, \dots, x_n) \in \lL$ it holds that $G \models_\lL \phi(v_1, \dots, v_n) \Leftrightarrow H \models_\lL \phi(w_1, \dots, w_n)$. Note that for labelled graphs, such an interpretation function is implicit: If the indices of $\free(\phi)$ are a subset of the labels of $G$, then we can interpret the variables $x_i$ by the vertex with the label $i$, that is, $\I(x_i) = \nu(i)$. The semantics of $\lC$ can then be stated succinctly in terms of label assignments.

\begin{definition}[$\lC$ semantics of labelled graphs]
  Let $\phi \in \lC$ and let $G$ be a labelled graph, such that $\nu(i) \in V(G)$ for all $x_i \in \free(\phi)$. Then $G \models \phi$ if
  \begin{itemize}
  \item $\phi = (x_i = x_j)$ and $\nu(i) = \nu(j)$
  \item $\phi = Ex_ix_j$ and $\nu(i)\nu(j) \in E(G)$
  \item $\phi = \psi \lor \theta$ and $G \models \psi$ or $G \models \theta$
  \item $\phi = \exists^{\geq t} x_\ell \psi(x_\ell)$ and there exist distinct $v_1, \dots, v_t$, such that $G(\ell \to v_i) \models \psi$ for all $i \in [t]$
  \end{itemize}
\end{definition}

Note that for labelled graphs this is equivalent to extending the standard semantics of $\lFO$ by the following rule: It is $G, v_1, \dots, v_n \models \exists^{\geq t}y \psi(x_1, \dots, x_n, y)$ if there exist distinct elements $u_1, \dots, u_t \in V(G)$ such that $G \models \psi(v_1, \dots, v_n, u_i)$ for all $i \in [t]$.

We sometimes write $\exists^{=t}x \phi(x)$ for $\exists^{\geq t} x \phi(x) \land \neg\exists^{\geq t+1}x \phi(x)$. We also write $\top$ for $\forall x (x=x)$ and $\bot$ for $\neg \top$.
As we already did above, we will often restrict ourselves to the connectives $\neg, \lor$ and the quantifier $\exists^{\geq t}$. This set of symbols is indeed equally expressive by De Morgan's laws and observing that $\exists x \phi(x) \equiv \exists^{\geq 1}x \phi(x)$ and $\forall x \phi(x) \equiv \neg\exists x \neg\phi(x)$.

The quantifier rank $\qr(\cdot)$ of a formula is defined inductively as follows. It is $\qr(\phi) = 0$ for atomic formulae $\phi$, $\qr(\neg\phi) = \qr(\phi)$, $\qr(\phi \lor \psi) = \max \{ \qr(\phi), \qr(\psi)\}$ and $\qr(\exists^{\geq t}x \phi) = 1 + \qr(\phi)$. The \emph{quantifier-rank-$q$ fragment} $\lC_q$ of counting first order logic consists of all formulae of quantifier rank at most $q$.

Instead of restricting the quantifer rank, we can also restrict the number of distinct variables that are allowed to occur in a formula. By $\lC^k$ we denote the \emph{$k$-variable fragment} of $\lC$, consisting of all formulae using at most $k$ different variables. Similarly, the \emph{$k$-variable quantifier-rank-$q$ fragment} is defined as $\lC^k_q \coloneqq \lC^k \cap \lC_q$. Note that these are purely syntactic definitions.

\subparagraph{Treewidth and Treedepth.}
Treewidth is a structural graph parameter that measures how close a graph is to being a tree. It is usually defined in terms of \emph{tree decompositions}.
\begin{definition}\label{def:treedecomp}
  A tree decomposition $(T, \beta)$ of a graph $G$ is a tree $T$ together with a function $\beta \colon V(T) \to 2^{V(G)}$ satisfying
  \begin{itemize}
  \item $\bigcup_{t \in V(T)} \beta(t) = V(G)$,
  \item for all $uv \in E(G)$ there is a $t \in V(T)$ with $u, v \in \beta(t)$,
  \item for all $v \in V(G)$ the set $\beta^{-1}(v) = \{t \in V(T) \mid v \in \beta(t)  \}$ is connected in $T$.
  \end{itemize}
\end{definition}

The sets $\beta(t)$ are called \emph{bags}. The \emph{width} of a tree decomposition is equal to $\max_{t \in V(T)}|\beta(t)|-1$. The treewidth $\tw(G)$ of a graph $G$ is the minimum width over all tree decompositions of $G$. We denote the class of graphs of treewidth at most $k$ by $\TW_k$. For an example of a tree decomposition see \cref{fig:treewidthtrees}.

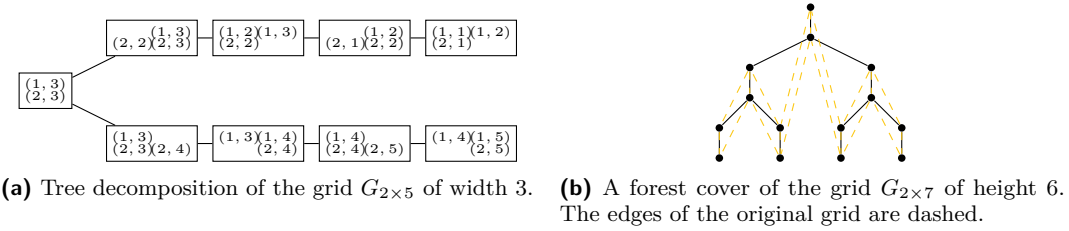
\begin{figure}
	\begin{subfigure}[t]{0.5\textwidth}
  \centering
  \begin{tikzpicture}[nodes={draw, minimum size=.25cm, inner sep=1pt}, scale=.7]
    \node at (0, 0) (v1) {
      \begin{tikzpicture}[scale=.1, node font=\tiny]
        \node[draw=none] at (0, 0) {$(1,3)$};
        \node[draw=none] at (0, -1.5) {$(2,3)$};
      \end{tikzpicture}
    };

    \node at (2, 1) (v21) {
      \begin{tikzpicture}[scale=.1, node font=\tiny]
      \node[draw=none] at (0, 0) {$(1,3)$};
      \node[draw=none] at (0, -1.5) {$(2,3)$};
      \node[draw=none] at (-5, -1.5) {$(2,2)$};
      \end{tikzpicture}
    };
	\node at (2, -1) (v22) {
		\begin{tikzpicture}[scale=.1, node font=\tiny]
		\node[draw=none] at (0, 0) {$(1,3)$};
		\node[draw=none] at (0, -1.5) {$(2,3)$};
		\node[draw=none] at (5, -1.5) {$(2,4)$};
		\end{tikzpicture}
	};

	\node at (4, 1) (v31) {
		\begin{tikzpicture}[scale=.1, node font=\tiny]
		\node[draw=none] at (0, 0) {$(1,3)$};
		\node[draw=none] at (-5,0) {$(1,2)$};
		\node[draw=none] at (-5, -1.5) {$(2,2)$};
		\end{tikzpicture}
	};
	\node at (4, -1) (v32) {
		\begin{tikzpicture}[scale=.1, node font=\tiny]
		\node[draw=none] at (0, 0) {$(1,3)$};
		\node[draw=none] at (5,0) {$(1,4)$};
		\node[draw=none] at (5, -1.5) {$(2,4)$};
		\end{tikzpicture}
	};

	\node at (6, 1) (v41) {
		\begin{tikzpicture}[scale=.1, node font=\tiny]
		\node[draw=none] at (0, 0) {$(1,2)$};
		\node[draw=none] at (0, -1.5) {$(2,2)$};
		\node[draw=none] at (-5, -1.5) {$(2,1)$};
		\end{tikzpicture}
	};
	\node at (6, -1) (v42) {
		\begin{tikzpicture}[scale=.1, node font=\tiny]
		\node[draw=none] at (0, 0) {$(1,4)$};
		\node[draw=none] at (0, -1.5) {$(2,4)$};
		\node[draw=none] at (5, -1.5) {$(2,5)$};
		\end{tikzpicture}
	};

	\node at (8, 1) (v51) {
		\begin{tikzpicture}[scale=.1, node font=\tiny]
		\node[draw=none] at (0, 0) {$(1,2)$};
		\node[draw=none] at (-5,0) {$(1,1)$};
		\node[draw=none] at (-5, -1.5) {$(2,1)$};
		\end{tikzpicture}
	};
	\node at (8, -1) (v52) {
		\begin{tikzpicture}[scale=.1, node font=\tiny]
		\node[draw=none] at (0, 0) {$(1,4)$};
		\node[draw=none] at (5,0) {$(1,5)$};
		\node[draw=none] at (5, -1.5) {$(2,5)$};
		\end{tikzpicture}
	};

    \draw (v51) -- (v41) -- (v31) -- (v21) -- (v1) -- (v22) -- (v32) -- (v42) -- (v52);

  \end{tikzpicture}
  \caption{Tree decomposition of the grid $\grid{2}{5}$ of width $3$.}
  \label{fig:treewidthtrees}
\end{subfigure}
 	~
	\begin{subfigure}[t]{0.47\textwidth}
	\centering
	\begin{tikzpicture}[scale=0.4,smallVertex/.style={fill=black, inner sep=1, circle}]
  		\node[smallVertex] (a03) at (0,-1) {};
  		\node[smallVertex] (a13) at (0,-2) {};
  		\node[smallVertex] (a01) at (-2,-3) {};
  		\node[smallVertex] (a11) at (-2,-4) {};
  		\node[smallVertex] (a05) at (2,-3) {};
  		\node[smallVertex] (a15) at (2,-4) {};
  		\node[smallVertex] (a00) at (-3,-5) {};
  		\node[smallVertex] (a10) at (-3,-6) {};
  		\node[smallVertex] (a02) at (-1,-5) {};
  		\node[smallVertex] (a12) at (-1,-6) {};
  		\node[smallVertex] (a06) at (3,-5) {};
  		\node[smallVertex] (a16) at (3,-6) {};
  		\node[smallVertex] (a04) at (1,-5) {};
  		\node[smallVertex] (a14) at (1,-6) {};
  		
  		\draw (a10) -- (a00) -- (a11) -- (a02) -- (a12);
  		\draw (a16) -- (a06) -- (a15) -- (a04) -- (a14);
  		\draw (a11) -- (a01) -- (a13) -- (a05) -- (a15);
  		\draw (a03) -- (a13);
  		
	    \foreach \i in {0,1,2,3,4,5,6}{
	    	\foreach \j in {0,1}{
	    		\ifthenelse{\NOT \i=0}{
	    			\tikzmath{
	    				integer \l;
	    				\l = \i - 1;}
	    			\draw[dashed, rwth-orange] (a\j\i) -- (a\j\l);
	    		}{}
	    	}
	    	\draw[dashed, rwth-orange] (a0\i) -- (a1\i);
	    }
	\end{tikzpicture}
	\caption{A  forest cover of the grid $\grid{2}{7}$ of height $6$. The edges of the original grid are dashed.}
	\label{fig:treedepthexample}
\end{subfigure}
 	\caption{Tree decomposition and forest cover of grids.}
	\label{fig:grid-decompositions}
\end{figure}

Treedepth, on the other hand, can be thought of as measuring how close a graph is to being a star. Alternatively, we may think of it as extending the notion of height beyond rooted forests. It is defined for a graph $G$ as the minimum height of a forest $F$ over the vertices of $G$, such that all edges in $G$ have an ancestor-descendant relationship in $F$.

\begin{definition}
  A \emph{forest cover} of a graph $G$ is a rooted forest $(F, \tup{r})$ with $V(F) = V(G)$, such that for every edge $uv \in E(G)$ it is either $u \preceq v$ or $v \preceq u$.
\end{definition}

The treedepth $\td(G)$ of $G$ is the minimum height of a forest cover of $G$. We denote the class of all graphs of treedepth at most $q$ by $\TD_q$. For an example of a forest cover see \cref{fig:treedepthexample}.

It is possible to construct a tree decomposition from a forest cover $(F, \tup{r})$. This is achieved by considering a path of bags, each containing the vertices on a path from $\tup{r}$ to the leaves of $F$. It is not hard to see that there is an ordering of these bags that satisfies the conditions of \cref{def:treedecomp}.
This yields the following relation between treedepth and treewidth.

\begin{fact}
  For every graph $G$, it holds that $\tw(G) \leq \td(G) - 1$.
\end{fact}

Both treewidth and treedepth enjoy characterisations in terms of node searching games, the so called \emph{cops-and-robber games}.
The general cops-and-robber game is played on a graph $G$ by Cops, controlling a number of cops; and Robber, controlling a single robber.
The cops and the robber are positioned on vertices of $G$.
The goal of Cops is to place a cop on the robber's position, while the robber tries to avoid capture by moving along paths free from cops.
The players play in rounds where first Cops announces the next position(s) of the cops (with possible restriction on how many cops may be moved and where they may be positioned) and then Robber moves the robber along some path avoiding all vertices where before and after his move there is a cop.
Treewidth can be characterised by the minimum number of cops needed to capture the robber where neither the movement of Cops nor Robber is restricted (see e.g. \cite{seymour_graph_1993}).
A well-known characterisation of treedepth is the minimum number of cops needed if Cops is not allowed to move a cop after it is positioned on the graph (see e.g. \cite{giannopoulou_lifo-search_2012}).
It is equivalent to count the number of rounds the game is played, without restricting the number of cops that can be used by Cops, as long as only one cop can be moved per round.
Therefore we use the following unified definition.

\begin{definition}[$q$-round $k$-cops-and-robber game]
	\label{def:Ekq-cops}
	Let $G$ be a graph and let $k,q\geq 1$.
	The \emph{monotone $q$-round $k$-cops-and-robber game} $\monCR_q^{k}(G)$ is defined as follows:
	
	We play the game on $G'$ which is constructed from $G$ by adding a disjoint $k$-vertex clique~$K$.
	The cop positions are sets $X\in \binom{V(G')}k$, the robber position is a vertex $v\in G$.
	The game is initiated with all cops positioned on $K$ and the robber on any vertex of $G$ of his choice.
	If the cops are at positions $X$ and robber at vertex $v$ we write $(X,v)$ for the position of the game.
	
	In round $i\leq q$, 
	\begin{itemize}
		\item Cops can move from the set $X_i$ to a set $X_{i+1}$ if $|X_i\cap X_{i+1}|=k-1$ and $\gamma_{v}^{X_{i}}\supseteq \gamma_{v}^{X_i\cap X_{i+1}}$,
		\item Robber moves along some path $v_iPv_{i+1}$, where no (inner) vertex is in $X_i\cap X_{i+1}$.
		\item Cops wins if $v_{i+1}\in X_{i+1}$.
	\end{itemize}
	Robber wins if Cops has not won after $q$ rounds.
	
	If we drop the condition $\gamma_{v}^{X_{i}}\supseteq \gamma_{v}^{X_i\cap X_{i+1}}$ in the movement of the cop, we call this the \emph{non-monotone} variant of the game and write $\CR_q^{k}(G)$.
\end{definition}

For $X\subseteq V(G')$ and $v\in G$, we call the connected component $\gamma^X_v$ of the graph $G\setminus X$, with $v\in \gamma^X_v$, the \emph{robber escape space}.
If the cop strategy only depends on $\gamma^X_v$ not the precise vertex that the robber occupies, we write $(X,\gamma^X_v)$ for the position of the game.

We write $\CR_q(G)$ instead of $\CR_q^q(G)$ and $\CR^k(G)$ instead of $\CR_{|V(G)|}^k(G)$.
Treewidth and treedepth can be expressed in terms of the existence of winning strategies in these games.

\begin{lemma}[{\cite[Theorem~1.4]{seymour_graph_1993}} and {\cite[Theorem~4]{giannopoulou_lifo-search_2012}}] \label{lem:games}
  Let $G$ be a graph. Let $k,q \geq 1$.
  \begin{enumerate}
  \item $G$ has treewidth at most $k$ iff Cops has a winning strategy for $\CR^{k+1}(G)$ iff Cops has a winning strategy for $\monCR^{k+1}(G)$.
  \item $G$ has treedepth at most $q$ iff Cops player has a winning strategy for $\CR_q(G)$ iff Cops has a winning strategy for $\monCR_q(G)$.
  \end{enumerate}
\end{lemma}

\section{Graph Decompositions Accounting for Treewidth and Treedepth Simultaneously}
\label{sec:graph-dec}
In this section, we reconcile treewidth and treedepth by introducing graph decompositions which account simultaneously for depth and width.
These efforts yield various equivalent characterisations of the graph class $\Ekq$, a subclass both of $\mathcal{TW}_{k-1}$ and $\mathcal{TD}_q$, the classes of graphs of treewidth $\leq k -1$ and treedepth $\leq q$, respectively.
By introducing a variant of the standard cops-and-robber game which captures $\Ekq$ and adopting a result from \cite{Furer01}, we show that $\Ekq$ is a proper subclass of $\mathcal{TW}_{k-1} \cap \mathcal{TD}_q$ if $q$ is sufficiently larger than $k$.

We start with the original definition of the class \Ekq, which incorporates width into forest covers from treedepth.
This definition has first been introduced as $k$-traversal in \cite{AbramskyDW17}.

\begin{definition} \label{def:kpebbleforest}
	Let $G$ be  graph and $k\geq 1$.
	A \emph{$k$-pebble forest cover} of $G$ is a tuple $(F,\vec{r},p)$, where $(F,\vec{r})$ is a rooted forest over the vertices $V(G)$ and a pebbling function $p\colon V(G)\rightarrow [k]$ such that:
	\begin{itemize}
		\item If $uv\in E(G)$, then $u\preceq v$ or $v\preceq u$ in $(F,\vec{r})$.
		\item If $uv \in E(G)$ and $u\prec v$ in $(F,\vec{r})$, then for every $w\in V(G)$ with $u\prec w\preceq v$ in $(F,\vec{r})$ it holds that $p(u)\neq p(w)$.
	\end{itemize}
	$(F,\vec{r},p)$ has depth $q\geq 1$ if $(F,\vec{r})$ has height $q$.
	We write $\Ekq$ for the class of all graphs $G$ admitting a $k$-pebble forest cover of depth $q$.
\end{definition}

The class $\Ekq$ can also be defined by measuring the depth of a tree decomposition $(T, \beta)$.
Crucially, it does not suffice to take the height of $T$ into account since this notion is not robust.
For example, it is well known that one can alter a tree decomposition by subdividing any edge multiple times and copying the bag of the child node.
This transformation does neither change the width of the decomposition, nor does it affect the information how to decompose the graph. However, the height of the tree will change drastically under this transformation.
It turns out that the following is the right definition:

\begin{definition}
	\label{def:depth-treedec}
	Let $G$ be a graph.
	A tuple $(T,r,\beta)$ is a \emph{rooted tree decomposition} of $G$
	if $(T,\beta)$ is a tree decomposition of $G$ and $r \in V(T)$.
	The \emph{depth of a tree decomposition $(T,\beta)$} is
	\begin{equation*}
		\dep(T,\beta)\coloneqq \min_{r\in V(T)} \dep(T,r,\beta)
		\qquad \text{with}
		\qquad
		\dep(T,r,\beta)\coloneqq \max_{v \in V(T)} \left| \bigcup_{t\preceq v} \beta(t) \right|.
	\end{equation*}
\end{definition}

Lastly we define a construction inspired by Dvo\v{r}\'ak~\cite{dvorak_recognizing_2010}, that enables us to use their proof technique to study the expressive power of first-order logic with counting quantifiers using homomorphism indistinguishability (see \cref{fig:elim-example}).

\begin{definition}
	\label{def:Ekq}
	Let $G$ be a (possibly labelled) graph. A \emph{$k$-construction tree} for $G$ is a tuple $(T,\lambda,r)$, where $T$ is a tree rooted at $r$ and $\lambda\colon V(T)\impl \CG_k$ is a function assigning $k$-labelled graphs to the nodes of $T$ such that:
	\begin{enumerate}
		\item $\lambda(r)=G$, \label{ct1}
		\item all leaves $\ell\in V(T)$ are assigned fully labelled graphs, that is $V(\lambda(\ell))=\labels{\lambda(\ell)}$,
		\item all internal nodes $t\in V(T)$ with exactly one child $t'$ are \emph{elimination nodes}, that is $\lambda(t)$ can be obtained from $\lambda(t')$ by removing one label, and
		\item all internal nodes $t\in V(T)$ with more than one child are \emph{product nodes}, that is $\lambda(t)$ is the product of its children.
	\end{enumerate}
	The \emph{\elimDepth} of a construction tree $(T,\lambda,r)$ is the maximum number of elimination nodes on any path from the root $r$ to a leaf.
	If $G$ has a $k$-construction tree of \elimDepth\ $\leq q$, we say that $G$ is \emph{\elimOrd{k}{q}}.
	For the class of all $k$-labelled \elimOrd{k}{q} graphs we write \Lkq.
\end{definition}

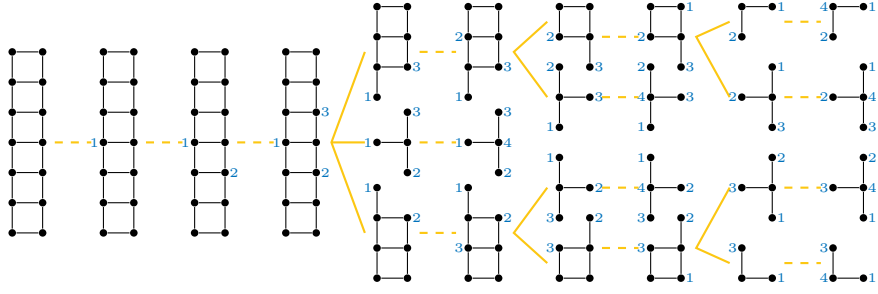
\begin{figure}
	\centering
	\begin{tikzpicture}[scale=0.4,smallVertex/.style={fill=black, inner sep=1, circle}, node font=\tiny]
	
\foreach \i in {0,1,2,3,4,5,6}{
		\foreach \j in {0,1}{
			\tikzmath{
				integer \newj;
				\newj = \j-3;}
			\node[smallVertex] (b\j\i) at (\newj,\i) {};
			\ifthenelse{\NOT \i=0}{
				\tikzmath{
					integer \l;
					\l = \i - 1;}
				\draw (b\j\i) -- (b\j\l);
			}{}
		}
		\draw (b0\i) -- (b1\i);
	}
	
\foreach \i in {0,1,2,3,4,5,6}{
		\foreach \j in {0,1}{
			\node[smallVertex] (a\j\i) at (\j,\i) {};
			\ifthenelse{\NOT \i=0}{
				\tikzmath{
					integer \l;
					\l = \i - 1;}
				\draw (a\j\i) -- (a\j\l);
			}{}
		}
		\draw (a0\i) -- (a1\i);
	}
	\node[rwth-blue] at (-0.3,3) {$1$};
	
\foreach \i in {0,1,2,3,4,5,6}{
		\foreach \j in {0,1}{
			\tikzmath{
				integer \newj;
				\newj = \j+3;}
			\node[smallVertex] (b\j\i) at (\newj,\i) {};
			\ifthenelse{\NOT \i=0}{
				\tikzmath{
					integer \l;
					\l = \i - 1;}
				\draw (b\j\i) -- (b\j\l);
			}{}
		}
		\draw (b0\i) -- (b1\i);
	}
	\node[rwth-blue] at (2.7,3) {$1$};
	\node[rwth-blue] at (4.3,2) {$2$};
	
\foreach \i in {0,1,2,3,4,5,6}{
		\foreach \j in {0,1}{
			\tikzmath{
				integer \newj;
				\newj = \j+6;}
			\node[smallVertex] (b\j\i) at (\newj,\i) {};
			\ifthenelse{\NOT \i=0}{
				\tikzmath{
					integer \l;
					\l = \i - 1;}
				\draw (b\j\i) -- (b\j\l);
			}{}
		}
		\draw (b0\i) -- (b1\i);
	}
	\node[rwth-blue] at (5.7,3) {$1$};
	\node[rwth-blue] at (7.3,2) {$2$};
	\node[rwth-blue] at (7.3,4) {$3$};
	
\foreach \i in {0,1,2}{
		\foreach \j in {0,1}{
			\tikzmath{
				integer \newj;
				\newj = \j+9;
				\newi = \i+5.5;}
			\node[smallVertex] (c\j\i) at (\newj,\newi) {};
			\ifthenelse{\NOT \i=0}{
				\tikzmath{
					integer \l;
					\l = \i - 1;}
				\draw (c\j\i) -- (c\j\l);
			}{}
		}
		\draw (c0\i) -- (c1\i);
	}
	\node[smallVertex] (c) at (9,4.5) {};
	\draw (c) -- (c00);
	\node[rwth-blue] at (8.7,4.5) {$1$};
	\node[rwth-blue] at (10.3,5.5) {$3$};
	
\foreach \i in {0,1,2}{
		\foreach \j in {0,1}{
			\tikzmath{
				integer \newj;
				\newj = \j+9;
				\newi = \i-1.5;}
			\node[smallVertex] (d\j\i) at (\newj,\newi) {};
			\ifthenelse{\NOT \i=0}{
				\tikzmath{
					integer \l;
					\l = \i - 1;}
				\draw (d\j\i) -- (d\j\l);
			}{}
		}
		\draw (d0\i) -- (d1\i);
	}
	\node[smallVertex] (d) at (9,1.5) {};
	\draw (d) -- (d02);
	\node[rwth-blue] at (8.7,1.5) {$1$};
	\node[rwth-blue] at (10.3,0.5) {$2$};
	
\node[smallVertex] (e1) at (9,3) {};
	\node[smallVertex] (e2) at (10,3) {};
	\node[smallVertex] (e3) at (10,4) {};
	\node[smallVertex] (e4) at (10,2) {};
	\draw (e1) -- (e2) -- (e3) (e2) -- (e4);
	\node[rwth-blue] at (8.7,3) {$1$};
	\node[rwth-blue] at (10.3,2) {$2$};
	\node[rwth-blue] at (10.3,4) {$3$};
	
\foreach \i in {0,1,2}{
		\foreach \j in {0,1}{
			\tikzmath{
				integer \newj;
				\newj = \j+12;
				\newi = \i+5.5;}
			\node[smallVertex] (c\j\i) at (\newj,\newi) {};
			\ifthenelse{\NOT \i=0}{
				\tikzmath{
					integer \l;
					\l = \i - 1;}
				\draw (c\j\i) -- (c\j\l);
			}{}
		}
		\draw (c0\i) -- (c1\i);
	}
	\node[smallVertex] (c) at (12,4.5) {};
	\draw (c) -- (c00);
	\node[rwth-blue] at (11.7,6.5) {$2$};
	\node[rwth-blue] at (11.7,4.5) {$1$};
	\node[rwth-blue] at (13.3,5.5) {$3$};
	
\foreach \i in {0,1,2}{
		\foreach \j in {0,1}{
			\tikzmath{
				integer \newj;
				\newj = \j+12;
				\newi = \i-1.5;}
			\node[smallVertex] (d\j\i) at (\newj,\newi) {};
			\ifthenelse{\NOT \i=0}{
				\tikzmath{
					integer \l;
					\l = \i - 1;}
				\draw (d\j\i) -- (d\j\l);
			}{}
		}
		\draw (d0\i) -- (d1\i);
	}
	\node[smallVertex] (d) at (12,1.5) {};
	\draw (d) -- (d02);
	\node[rwth-blue] at (11.7,1.5) {$1$};
	\node[rwth-blue] at (11.7,-0.5) {$3$};
	\node[rwth-blue] at (13.3,0.5) {$2$};
	
\node[smallVertex] (e1) at (12,3) {};
	\node[smallVertex] (e2) at (13,3) {};
	\node[smallVertex] (e3) at (13,4) {};
	\node[smallVertex] (e4) at (13,2) {};
	\draw (e1) -- (e2) -- (e3) (e2) -- (e4);
	\node[rwth-blue] at (11.7,3) {$1$};
	\node[rwth-blue] at (13.3,2) {$2$};
	\node[rwth-blue] at (13.3,4) {$3$};
	\node[rwth-blue] at (13.3,3) {$4$};
	
\foreach \i in {0,1}{
		\foreach \j in {0,1}{
			\tikzmath{
				integer \newj;
				\newj = \j+15;
				\newi = \i+6.5;}
			\node[smallVertex] (c\j\i) at (\newj,\newi) {};
			\ifthenelse{\NOT \i=0}{
				\tikzmath{
					integer \l;
					\l = \i - 1;}
				\draw (c\j\i) -- (c\j\l);
			}{}
		}
		\draw (c0\i) -- (c1\i);
	}
	\node[smallVertex] (c) at (16,5.5) {};
	\draw (c) -- (c10);
	\node[rwth-blue] at (14.7,6.5) {$2$};
	\node[rwth-blue] at (16.3,5.5) {$3$};
	
\foreach \i in {0,1}{
		\foreach \j in {0,1}{
			\tikzmath{
				integer \newj;
				\newj = \j+15;
				\newi = \i-1.5;}
			\node[smallVertex] (d\j\i) at (\newj,\newi) {};
			\ifthenelse{\NOT \i=0}{
				\tikzmath{
					integer \l;
					\l = \i - 1;}
				\draw (d\j\i) -- (d\j\l);
			}{}
		}
		\draw (d0\i) -- (d1\i);
	}
	\node[smallVertex] (d) at (16,0.5) {};
	\draw (d) -- (d11);
	\node[rwth-blue] at (14.7,-0.5) {$3$};
	\node[rwth-blue] at (16.3,0.5) {$2$};
	
\node[smallVertex] (e1) at (16,1.5) {};
	\node[smallVertex] (e2) at (15,1.5) {};
	\node[smallVertex] (e3) at (15,2.5) {};
	\node[smallVertex] (e4) at (15,0.5) {};
	\draw (e1) -- (e2) -- (e3) (e2) -- (e4);
	\node[rwth-blue] at (16.3,1.5) {$2$};
	\node[rwth-blue] at (14.7,0.5) {$3$};
	\node[rwth-blue] at (14.7,2.5) {$1$};
	
	\node[smallVertex] (e1) at (16,4.5) {};
	\node[smallVertex] (e2) at (15,4.5) {};
	\node[smallVertex] (e3) at (15,5.5) {};
	\node[smallVertex] (e4) at (15,3.5) {};
	\draw (e1) -- (e2) -- (e3) (e2) -- (e4);
	\node[rwth-blue] at (16.3,4.5) {$3$};
	\node[rwth-blue] at (14.7,3.5) {$1$};
	\node[rwth-blue] at (14.7,5.5) {$2$};
	
\foreach \i in {0,1}{
		\foreach \j in {0,1}{
			\tikzmath{
				integer \newj;
				\newj = \j+18;
				\newi = \i+6.5;}
			\node[smallVertex] (c\j\i) at (\newj,\newi) {};
			\ifthenelse{\NOT \i=0}{
				\tikzmath{
					integer \l;
					\l = \i - 1;}
				\draw (c\j\i) -- (c\j\l);
			}{}
		}
		\draw (c0\i) -- (c1\i);
	}
	\node[smallVertex] (c) at (19,5.5) {};
	\draw (c) -- (c10);
	\node[rwth-blue] at (17.7,6.5) {$2$};
	\node[rwth-blue] at (19.3,5.5) {$3$};
	\node[rwth-blue] at (19.3,7.5) {$1$};
	
\foreach \i in {0,1}{
		\foreach \j in {0,1}{
			\tikzmath{
				integer \newj;
				\newj = \j+18;
				\newi = \i-1.5;}
			\node[smallVertex] (d\j\i) at (\newj,\newi) {};
			\ifthenelse{\NOT \i=0}{
				\tikzmath{
					integer \l;
					\l = \i - 1;}
				\draw (d\j\i) -- (d\j\l);
			}{}
		}
		\draw (d0\i) -- (d1\i);
	}
	\node[smallVertex] (d) at (19,0.5) {};
	\draw (d) -- (d11);
	\node[rwth-blue] at (17.7,-0.5) {$3$};
	\node[rwth-blue] at (19.3,0.5) {$2$};
	\node[rwth-blue] at (19.3,-1.5) {$1$};
	
\node[smallVertex] (e1) at (19,1.5) {};
	\node[smallVertex] (e2) at (18,1.5) {};
	\node[smallVertex] (e3) at (18,2.5) {};
	\node[smallVertex] (e4) at (18,0.5) {};
	\draw (e1) -- (e2) -- (e3) (e2) -- (e4);
	\node[rwth-blue] at (19.3,1.5) {$2$};
	\node[rwth-blue] at (17.7,0.5) {$3$};
	\node[rwth-blue] at (17.7,2.5) {$1$};
	\node[rwth-blue] at (17.7,1.5) {$4$};
	
	\node[smallVertex] (e1) at (19,4.5) {};
	\node[smallVertex] (e2) at (18,4.5) {};
	\node[smallVertex] (e3) at (18,5.5) {};
	\node[smallVertex] (e4) at (18,3.5) {};
	\draw (e1) -- (e2) -- (e3) (e2) -- (e4);
	\node[rwth-blue] at (19.3,4.5) {$3$};
	\node[rwth-blue] at (17.7,3.5) {$1$};
	\node[rwth-blue] at (17.7,5.5) {$2$};
	\node[rwth-blue] at (17.7,4.5) {$4$};
	
\node[smallVertex] (c1) at (21,6.5) {};
	\node[smallVertex] (c2) at (21,7.5) {};
	\node[smallVertex] (c3) at (22,7.5) {};
	\draw (c1) -- (c2) -- (c3);
	\node[rwth-blue] at (20.7,6.5) {$2$};
	\node[rwth-blue] at (22.3,7.5) {$1$};
	
	\node[smallVertex] (d1) at (21,-0.5) {};
	\node[smallVertex] (d2) at (21,-1.5) {};
	\node[smallVertex] (d3) at (22,-1.5) {};
	\draw (d1) -- (d2) -- (d3);
	\node[rwth-blue] at (20.7,-0.5) {$3$};
	\node[rwth-blue] at (22.3,-1.5) {$1$};
	
\node[smallVertex] (e1) at (21,1.5) {};
	\node[smallVertex] (e2) at (22,1.5) {};
	\node[smallVertex] (e3) at (22,2.5) {};
	\node[smallVertex] (e4) at (22,0.5) {};
	\draw (e1) -- (e2) -- (e3) (e2) -- (e4);
	\node[rwth-blue] at (20.7,1.5) {$3$};
	\node[rwth-blue] at (22.3,0.5) {$1$};
	\node[rwth-blue] at (22.3,2.5) {$2$};
	
	\node[smallVertex] (e1) at (21,4.5) {};
	\node[smallVertex] (e2) at (22,4.5) {};
	\node[smallVertex] (e3) at (22,5.5) {};
	\node[smallVertex] (e4) at (22,3.5) {};
	\draw (e1) -- (e2) -- (e3) (e2) -- (e4);
	\node[rwth-blue] at (20.7,4.5) {$2$};
	\node[rwth-blue] at (22.3,3.5) {$3$};
	\node[rwth-blue] at (22.3,5.5) {$1$};
	
\node[smallVertex] (c1) at (24,6.5) {};
	\node[smallVertex] (c2) at (24,7.5) {};
	\node[smallVertex] (c3) at (25,7.5) {};
	\draw (c1) -- (c2) -- (c3);
	\node[rwth-blue] at (23.7,6.5) {$2$};
	\node[rwth-blue] at (25.3,7.5) {$1$};
	\node[rwth-blue] at (23.7,7.5) {$4$};
	
	\node[smallVertex] (d1) at (24,-0.5) {};
	\node[smallVertex] (d2) at (24,-1.5) {};
	\node[smallVertex] (d3) at (25,-1.5) {};
	\draw (d1) -- (d2) -- (d3);
	\node[rwth-blue] at (23.7,-0.5) {$3$};
	\node[rwth-blue] at (25.3,-1.5) {$1$};
	\node[rwth-blue] at (23.7,-1.5) {$4$};
	
\node[smallVertex] (e1) at (24,1.5) {};
	\node[smallVertex] (e2) at (25,1.5) {};
	\node[smallVertex] (e3) at (25,2.5) {};
	\node[smallVertex] (e4) at (25,0.5) {};
	\draw (e1) -- (e2) -- (e3) (e2) -- (e4);
	\node[rwth-blue] at (23.7,1.5) {$3$};
	\node[rwth-blue] at (25.3,0.5) {$1$};
	\node[rwth-blue] at (25.3,2.5) {$2$};
	\node[rwth-blue] at (25.3,1.5) {$4$};
	
	\node[smallVertex] (e1) at (24,4.5) {};
	\node[smallVertex] (e2) at (25,4.5) {};
	\node[smallVertex] (e3) at (25,5.5) {};
	\node[smallVertex] (e4) at (25,3.5) {};
	\draw (e1) -- (e2) -- (e3) (e2) -- (e4);
	\node[rwth-blue] at (23.7,4.5) {$2$};
	\node[rwth-blue] at (25.3,3.5) {$3$};
	\node[rwth-blue] at (25.3,5.5) {$1$};
	\node[rwth-blue] at (25.3,4.5) {$4$};
	
\draw[dashed, thick, rwth-orange] (-1.6,3) -- (-0.4,3) (1.4,3) -- (2.6,3) (4.4,3) -- (5.6,3) (10.4,6) -- (11.6,6) (10.4,3) -- (11.6,3) (10.4,0) -- (11.6,0) (16.4,6.5) -- (17.6,6.5) (16.4,4.5) -- (17.6,4.5) (16.4,1.5) -- (17.6,1.5) (16.4,-0.5) -- (17.6,-0.5) (22.4,7) -- (23.6,7) (22.4,4.5) -- (23.6,4.5) (22.4,1.5) -- (23.6,1.5) (22.4,-1) -- (23.6,-1);
	\draw[thick, rwth-orange] (8.6,0) -- (7.5,3) -- (8.6,6) (7.5,3) -- (8.6,3) (14.6,4.5) -- (13.5,6) -- (14.6,7) (14.6,-1) -- (13.5,0) -- (14.6,1.5) (20.6,4.5) -- (19.5,6.5) -- (20.6,7) (20.6,-1) -- (19.5,-0.5) -- (20.6,1.5);
	\end{tikzpicture}
	\caption{A $4$-construction tree for the grid $\grid{2}{7}$ of elimination depth $6$. Edges entering elimination nodes are dashed.}
	\label{fig:elim-example}
\end{figure}

It turns out that all three notions coincide.

\begin{restatable}{theorem}{ekqEquiv}
	\label{thm:Ekq_equiv}
	Let $k, q\geq 1$. For every graph $G$, the following are equivalent:
	\begin{enumerate}
		\item $G$ is \elimOrd{k}{q},\label{ekqEquiv1}
		\item $G$ has a tree decomposition of width $k-1$ and depth $q$,\label{ekqEquiv2}
		\item $G\in\Ekq$, that is $G$ admits a $k$-pebble forest cover of depth $q$.\label{ekqEquiv3}
	\end{enumerate}
\end{restatable}

The equivalence of \cref{ekqEquiv1,ekqEquiv2} is proven by a careful choice of the tree decomposition where one than can identify the bags with the labelled vertices of the construction tree. For the equivalence of \cref{ekqEquiv2,ekqEquiv3}, we follow the proof of \cite[Theorem~19]{abramsky_relating_2021} and observe that their construction preserves depth. Details can be found in \cref{app:equiv-classes}.

\begin{corollary}
	\label{cor:Ekq-subset-tw-td}
	\label{cor:minor-closed}
	Let $k, q\geq 1$.
	$\Ekq $ is minor-closed, closed under taking disjoint unions, and a subclass of $\TW_{k-1}\cap\TD_q$.
\end{corollary}

Given \cref{thm:Ekq_equiv}, \cref{cor:minor-closed} is immediate.
Dawar, Jakl, and Reggio reduced the proofs of the results of Grohe and Dvo\v{r}\'ak \cite{grohe_counting_2020,dvorak_recognizing_2010} to a ``combinatorial core'' \cite[Remark~17]{dawar_lovasz-type_2021}, which amounts to showing that the classes $\mathcal{TW}_k$ and $\mathcal{TD}_q$ are closed under contracting edges.
To that end, \cref{cor:minor-closed} illustrates the benefits of characterising $\Ekq$ in terms of tree decompositions (\cref{def:depth-treedec}): Proving that pebble forest covers are preserved under edge contractions requires a non-trivial amount of bookkeeping while the analogous statement for tree decomposition is straightforward.

We conclude with a characterisation of $\Ekq$ in terms of a cops-and-robber game.

\begin{restatable}{lemma}{crGame}
	\label{lem:Ekq-cops}
	The Cops win the game $\monCR_q^{k}(G)$ if and only if $G\in \Ekq$.
\end{restatable}

The proof of this lemma follows the same strategies as the proof for the monotone version of the cops-and-robber game for treewidth (see for example \cite{Rabinovich14}). Details can be found in \cref{app:cr-equiv}.

\subsection{Separating $\Ekq$ from $\mathcal{TW}_k \cap \mathcal{TD}_q$ Syntactically}

We aim to show that the graph class $\Ekq$ is a proper subclass of $\mathcal{TW}_k \cap \mathcal{TD}_q$.
Since $\EParam{q}{q}=\TD_q=\TW_{q-1}\cap\TD_q$, one can only hope to separate the classes if $q$ is larger than $k$.
Using the characterisations of $\Ekq$ via a cops-and-robber game, we show that there indeed are graphs for which one can not simultaneously bound the width to treewidth and the depth to treedepth.
The graph we consider is the $(h\times\ell)$-grid \grid{h}{\ell}, with $h<\ell$. It is well known that $\tw(\grid{h}{\ell})=h$ and $\td(\grid{h}{\ell})\leq h \lceil\log(\ell+1)\rceil$, cf.\@ \cref{fig:grid-decompositions}. We give a lower bound to the number of rounds that the robber can survive in a, possibly non-monotone, game $\CR(G,h)$, which is linear in both $\ell$ and $h$.

\begin{restatable}{lemma}{lowerBoundRounds}
	\label{lem:lower-bound-rounds}
	For $1<h<\ell-2$ and $q\leq\frac{h(\ell  - h + 2)}{4}$, Robber wins the game $\CR_q^{h+1}(\grid{h}{\ell})$.
\end{restatable}

The proof of \cref{lem:lower-bound-rounds} builds upon \cite{Furer01}.
The winning strategy of the robber is to always stay in the component with the most vertices.
We find a lower bound on the size of this component in terms of the number of rounds played and prove that the cop player can only force this bound to shrink by two vertices each round, for the majority of the rounds.
We additionally show that for $h>3$ Cops can indeed force the component to shrink by two vertices each round and thus in this case the bound given in \cref{lem:lower-bound-rounds} is tight up to an additive term depending only on $h$.

For $h=1$, the proof idea of \cref{lem:lower-bound-rounds} is not applicable as on a path there are separators of size two that separate the path into three components of roughly equal size.
Despite that, one may observe that such a separator does not benefit Cops as from such a position he would always have to combine two of these components into a larger one.
Thus the cop player can only move along the path and shrink the escape space of the robber by one vertex.
This case is covered in the original proof of \cite{Furer01}.

\begin{lemma}[{\cite[Theorems~5~and~7]{Furer01}}]
	\label{lem:rounds-path}
	Let $\ell\geq 1$. Robber wins the game $\CR_q^2(\grid{1}{\ell})$ if and only if $q\leq\lceil \frac{\ell-1}{2}\rceil$.
\end{lemma}

With an appropriate choice of $\ell$ and a small alteration to the graph \grid{k-1}{\ell} that ensures that the treedepth of the graph is exactly $q$, we can prove the following:

\ekqtwtdsyntax*

Detailed proofs can be found in \cref{app:rounds}.
The reader should note that the proof of the lower bound on the number of rounds even holds for the non-monotone game, which in turn allows us to lift this result to the semantic level of homomorphism indistinguishability.
 
\section{Homomorphism Indistinguishability} \label{sec:homind}
In this section, we turn to investigating $\Ekq$ in terms of homomorphism indistinguishability.
It turns out that the representation of $\Ekq$ in terms of construction trees offers a great framework for obtaining characterisations of logical equivalence. The general idea will be to use these trees to inductively construct $\lC$-formulae that capture homomorphism counts. Not only does this approach generalise results from \cite{dvorak_recognizing_2010,grohe_counting_2020}, it also yields an intuitive characterisation of $\lC^k_q$-equivalence. This provides a more elementary proof of a result from \cite{dawar_lovasz-type_2021}.

Moreover, the constructive nature of our proof strategy proves fruitful in obtaining additional characterisations of fragments of $\lC$. The general idea is to impose natural restrictions on the construction trees, such that a fragment $\lL \subsetneq \lC$ already suffices to capture homomorphism counts. By choosing these restrictions carefully, the resulting subclass of $\Ekq$ is then still large enough to capture $\lL$-equivalence. We illustrate this point by giving a characterisation of \emph{guarded} counting logic $\lGC$.

We conclude by \emph{semantically} separating $\Ekq$ and $\mathcal{TW}_{k-1} \cap \mathcal{TD}_q$. More formally, we show that, for $q$ sufficiently larger than $k$, there exist graphs $G$ and $H$ which are homomorphism indistinguishable over $\Ekq$ but have different numbers of homomorphisms from some graph in $\mathcal{TW}_{k-1} \cap \mathcal{TD}_q$.

\subsection{Homomorphism Indistinguishability over $\Ekq$ is $\lC^k_q$-Equivalence}
\label{subsec:ckq}
In his 2010 paper \cite{dvorak_recognizing_2010}, Dvo\v{r}ák showed that $\lC^k$-equivalence is equivalent to homomorphism indistinguishability over $\TW_k$. It turns out that his techniques generalize remarkably well to construction trees.
To begin with, we make a few observations on how the operations that make up construction trees interact with homomorphism counts.

\begin{observation}\label{prp:homproducts}
  For labelled graphs $F_1, F_2, G$, it holds that $\hom(F_1F_2, G) = \hom(F_1, G) \cdot \hom(F_2, G)$.
\end{observation}

This is because any two homomorphisms $g \colon F_1 \to G$ and $h \colon F_2 \to G$ must agree on vertices with the same label, so $g \sqcup h$ is well-defined and a homomorphism from $F_1F_2$ to $G$. Moreover, for $h \in \HOM(F_1F_2, G)$ the restrictions $h\restrict{V(F_1)}$ and $h\restrict{V(F_2)}$ are homomorphisms. We can also relate the homomorphism counts from graphs $F$ and $F'$, whenever $F'$ is obtained from $F$ by removing some label $\ell$. Then in any homomorphism $h \colon F' \to G$ the image of $\nu_F(\ell)$ is no longer necessarily $\nu_G(\ell)$. Hence, we can obtain $\hom(F, G)$ by moving the the label $\ell$ to different vertices in $G$ and tallying up the homomorphisms from $F$ to those graphs.

\begin{observation}\label{prp:homlabeldel}
  Let $F'$ be the graph obtained from $F$ by removing a single label $\ell$. Then $\hom(F', G) = m$ if and only if there exists a decomposition $m = \sum_{i=1}^tc_im_i$ with $c_i, m_i \in \N$, such that:
  \begin{itemize}
  \item There exist exactly $c_i$ vertices $v$ with $\hom(F, G(\ell \rightarrow v)) = m_i$.
    \item There exist exactly $c \coloneqq \sum_i c_i$ vertices $v$ with $\hom(F, G(\ell \rightarrow v)) \neq 0$.
  \end{itemize}
\end{observation}

Finally, observe that when $F$ is fully labelled there can be at most one homomorphism $h \colon F \to G$, which is entirely determined by the label positions in $G$.

\begin{observation}
  Let $F$ be a fully labelled graph and let $L_F$ denote the set of labels. Then there exists a unique homomorphism $h \colon F \to G$ if for all labels $i, j \in L_F$
  \begin{itemize}
  \item $\nu_F(i) = \nu_F(j) \implies \nu_G(i) = \nu_G(j)$,
  \item $\nu_F(i)\nu_F(j) \in E(F) \implies \nu_G(i)\nu_G(j) \in E(G)$.
  \end{itemize}
\end{observation}

The crucial insight is these conditions are all definable in $\lC$. The condition for fully labelled graphs in particular can be expressed as a conjunction of atomic formulae using at most $|L_F|$ different variables. This allows us to prove the following lemma by induction over a construction tree.
The proofs of the lemmas in this section can be found in \Cref{app:c-homind}.

\begin{restatable}{lemma}{homcountsInCkq}
  \label{lem:homcounts_in_ckq}
  Let $F \in \Lkq$ be a $k$-labelled graph, and let $m \geq 0$. Then there exists a formula $\phi_m \in \lC^k_q$ such that for each $k$-labelled graph $G$ with $L_F \subseteq L_G$, $G \models \phi_m$ if and only if $\hom(F, H) = m$.
\end{restatable}

Ideally, we would like to prove the converse in a similar manner. Given some $\lC^k_q$ formula $\psi$ that distinguishes two graphs $G$ and $H$, construct a graph $F \in \Lkq$ with $\hom(F, G) \neq \hom(F, H)$ per induction over the structure of $\psi$. While graphs are too rigid in this regard, such a construction will be possible using \emph{linear combinations} of graphs.\footnote{These linear combinations are called ``quantum graphs'' in \cite{dvorak_recognizing_2010}.}

For a class of (labelled) graphs $\mathcal{F}$, we let $\R\mathcal{F}$ be the class of finite formal linear combinations with real coefficients of graphs $F \in \mathcal{F}$. We generalise the $\hom$ function to $\R\CG$ in a linear way by defining
\begin{equation*}
    \hom(\qg{F}, G) = \hom(\sum_i c_iF_i, G) \coloneqq \sum_i c_i \cdot \hom(F_i, G).
  \end{equation*}
for $\qg{F} = \sum_i c_iF_i$.

  We make the following observation, which will allow us to reason about linear combinations instead of graphs.

  \begin{observation}\label{prp:qgdistinguishing}
    Let $G, H$ be graphs and let $\qg{F} \in \R\mathcal{F}$. If $\hom(\qg{F}, G) \neq \hom(\qg{F}, H)$, then there is already an $F \in \mathcal{F}$ with $\hom(F, G) \neq \hom(F, H)$.
  \end{observation}

  The product of two linear combinations is defined in the natural way, where the graph products distribute over the sum. We also remove any graphs with loops that might have been created from the resulting linear combination. This definition preserves the property that $\hom(\qg{F}_1\qg{F}_2, H) = \hom(\qg{F}_1, H)\hom(\qg{F}_2, H)$ and admits the following interpolation lemma.

  \begin{lemma}[{\cite[Lemma 5]{dvorak_recognizing_2010}}]
    \label{lem:qginterpolation}
  Let $\mathcal{F}$ be a class of graphs and let $\qg{F} \in \R\mathcal{F}$. If $S^-, S^+$ are disjoint finite sets of real numbers, then there exists a linear combination $\qgp{F}{S^+}{S^-}$, such that for any graph $G$
  \begin{itemize}
    \item $\hom(\qgp{F}{S^+}{S^-}, G) = 1$ if $\hom(\qg{F}, G) \in S^+$, and
    \item $\hom(\qgp{F}{S^+}{S^-}, G) = 0$ if $\hom(\qg{F}, G) \in S^-$.
  \end{itemize}
    Moreover, if $\mathcal{F}$ is closed under taking products then $\qgp{F}{S^+}{S^-} \in \R\mathcal{F}$.
  \end{lemma}

  With this result, we may construct for a formula $\phi \in \lC^k_q$ and $n \in \N$ a linear combination $\qg{F}_{\psi, n}$ such that for all graphs $G$ of size $n$ it holds that $\hom(\qg{F}_{\psi, n}, G) = 1$ if $G \models \psi$ and $\hom(\qg{F}_{\psi, n}, G) = 0$ otherwise. We say that $\qg{F}_{\psi, n}$ \emph{models $\psi$ for graphs of size $n$}.

  \begin{restatable}{lemma}{qgFromCkq}
    \label{lem:qg_from_ckq}
  Let $k, q \geq 1$ and let $\phi$ be a $\lC^k_q$-formula. Then for every $n \geq 1$ there exists an $\qg{F} \in \R\Lkq$ modelling $\phi$ for graphs of size $n$.
\end{restatable}

The proof is per induction over the structure of $\phi$ and exploits how homomorphism counts change under label deletions and taking products. The interpolation construction is used to define negation and to renormalise homomorphism counts to $0$ or $1$. The construction has the property that labels in the components of $\qg{F}$ correspond to free variables of $\phi$. This motivates the following corollary.

\begin{corollary}\label{cor:qg_from_ckq}
  Let $k, q \geq 1$ and let $\phi$ be a $\lC^k_q$-sentence. Then for every $n \geq 1$ there exists an $\qg{F} \in \R\Ekq$ modelling $\phi$ for graphs of size $n$.
\end{corollary}

  We can now prove the main result of this section.

  \begin{theorem}\label{thm:ckq_equivalence}
    Let $k,q \geq 1$. 
    Two graphs $G$ and $H$ are $\lC^k_q$-equivalent if and only if they are homomorphism indistinguishable over $\Ekq$.
  \end{theorem}
  \begin{proof}
    Suppose there was a graph $F \in \Ekq \subseteq \Lkq$ with $\hom(F, G) \neq \hom(F, H)$. Then by \cref{lem:homcounts_in_ckq} there exist $\lC^k_q$ sentences $\varphi^F_m$ such that $G \models \varphi^F_m$ iff $\hom(F, G) = m$. Consequently, there exists an $m$ with $G \models \varphi^F_m$ and $H \not\models \varphi^F_m$, so $G$ and $H$ cannot satisfy the same $\lC^k_q$ sentences.

    Suppose now there was a sentence $\phi \in \lC^k_q$ with $G \models \phi$ and $H \not\models \phi$. Wlog we assume $|G| = |H| = m$. Then by \cref{cor:qg_from_ckq} there is an $\qg{F} \in \R\Ekq$ that models $\phi$ for graphs of size $m$, that is, $\hom(\qg{F}, G) \neq \hom(\qg{F}, H)$. By \cref{prp:qgdistinguishing}, this already implies the existence of an $F \in \Ekq$ with $\hom(F, G) \neq \hom(F, H)$.
  \end{proof}

  By dropping the restriction on one of the parameters in \cref{thm:ckq_equivalence}, we recover the original results of Dvo\v{r}\'ak~\cite{dvorak_recognizing_2010} and Grohe~\cite{grohe_counting_2020}:

    \begin{corollary}
  	Let $k, q \geq 1$. Let $G$ and $H$ be graphs.
  	\begin{enumerate}
  		\item $G$ and $H$ are $\lC^k$-equivalent iff they are homomorphism indistinguishable over $\mathcal{TW}_{k-1}$.
  		\item $G$ and $H$ are $\lC_q$-equivalent iff they are homomorphism indistinguishable over $\mathcal{TD}_q$.
  	\end{enumerate}
      \end{corollary}

\subsection{Guarded Fragments}
\label{subsec:gckq}
Given the constructive nature of these proofs, it is interesting to investigate whether the same strategy may be used to obtain results for different fragments of $\lC$ by restricting construction trees in some way. An example where this works well is guarded counting logic $\lGC$.

In the guarded fragment $\lGC$, quantifiers are restricted to range over the neighbours of a vertex. Formally, we require that quantifiers only occur in the form $\exists^{\geq t}y (Exy \land \psi(z_1, \dots, z_n, y))$, where $x$ and $y$ are distinct variables.

Since $\lGC$-formulae necessarily have a free variable, it is not immediately obvious how to define $\lGC$-equivalence on graphs. One option is to consider graphs together with a distinguished vertex and write $G, v \equiv_\lGC H, w$. This, however, does not yield an equivalence relation on graphs alone. To that end, we choose the following natural definition.

\begin{definition}[\textsf{GC}-equivalence]\label{def:gcequiv}
  Let $G$ and $H$ be unlabelled graphs. We say that $G$ and $H$ are $\textsf{GC}$-equivalent, in symbols $G \equiv_{\mathsf{GC}} H$, if there exists a bijection $f \colon V(G) \to V(H)$ such that $G, v \models \varphi(x) \iff H, f(v) \models \varphi(x)$ for all $v \in V(G)$ and $\varphi \in \mathsf{GC}$.
\end{definition}

To apply our proof strategy to $\lGC$, we need to restrict the construction trees such that guarded quantifiers suffice to express homomorphism counts. Observe that the quantifiers are only needed to describe how the number of homomorphisms $F \to G$ changes by removing a label from $F$. More precisely, we use that removing a label $\ell$ from $F$ is the same as moving it around in $G$ and tallying up the resulting homomorphisms. Now if $\ell$ is adjacent to some other label $\ell'$, then the only positions of $\ell$ in $G$ that contribute to the final homomorphism count are adjacent to $\ell'$. Consequently, it will suffice to quantify over the neighbours of $\ell'$.

\begin{definition}
  Let $k, q \geq 1$. By $\GEkq$ we denote the class of $k$-labelled graphs that admit a $k$-construction tree of elimination-depth $q$ with the additional restriction that labels can only be removed if they have a labelled neighbour.
\end{definition}

We observe that in \cref{fig:elim-example} there are nodes where labels without labelled neighbors are removed. In \cref{fig:guard-elim-example}, we depict a construction tree without such nodes of the same graph. We remark that all graphs in $\GEkq$ are labelled, as a single label can never be removed.
Under these restrictions, the argument from \cref{lem:homcounts_in_ckq} goes through using only guarded quantifiers.

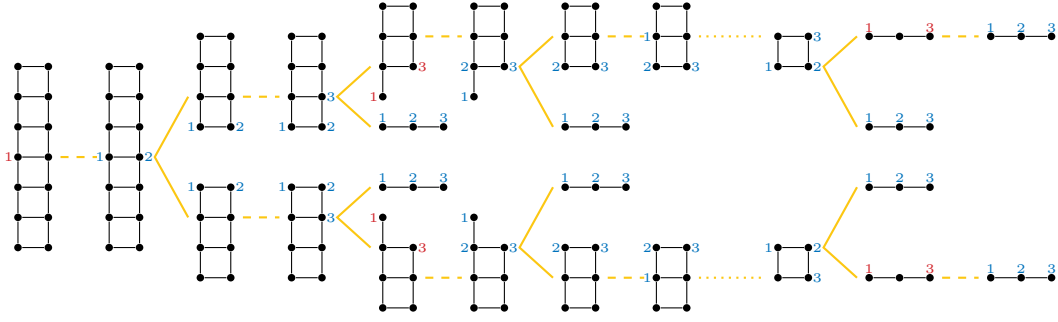
\begin{figure}
  \centering
  \begin{tikzpicture}[scale=0.4,smallVertex/.style={fill=black, inner sep=1, circle}, node font=\tiny]
\foreach \i in {0,1,2,3,4,5,6}{
    	\foreach \j in {0,1}{
    		\node[smallVertex] (a\j\i) at (\j,\i) {};
    	    \ifthenelse{\NOT \i=0}{
				\tikzmath{
					integer \l;
					\l = \i - 1;}
				\draw (a\j\i) -- (a\j\l);
			}{}
    	}
    	\draw (a0\i) -- (a1\i);
	}
	\node[rwth-red] at (-0.3,3) {$1$};
	
\foreach \i in {0,1,2,3,4,5,6}{
		\foreach \j in {0,1}{
			\tikzmath{
				integer \newj;
				\newj = \j+3;}
			\node[smallVertex] (b\j\i) at (\newj,\i) {};
			\ifthenelse{\NOT \i=0}{
				\tikzmath{
					integer \l;
					\l = \i - 1;}
				\draw (b\j\i) -- (b\j\l);
			}{}
		}
		\draw (b0\i) -- (b1\i);
	}
	\node[rwth-blue] at (2.7,3) {$1$};
	\node[rwth-blue] at (4.3,3) {$2$};
	
\foreach \i in {0,1,2,3}{
		\foreach \j in {0,1}{
			\tikzmath{
				integer \newj;
				\newj = \j+6;
				integer \newi;
				\newi = \i+4;}
			\node[smallVertex] (c\j\i) at (\newj,\newi) {};
			\ifthenelse{\NOT \i=0}{
				\tikzmath{
					integer \l;
					\l = \i - 1;}
				\draw (c\j\i) -- (c\j\l);
			}{}
		}
		\draw (c0\i) -- (c1\i);
	}
	\node[rwth-blue] at (5.7,4) {$1$};
	\node[rwth-blue] at (7.3,4) {$2$};
	
\foreach \i in {0,1,2,3}{
		\foreach \j in {0,1}{
			\tikzmath{
				integer \newj;
				\newj = \j+6;
				integer \newi;
				\newi = \i-1;}
			\node[smallVertex] (d\j\i) at (\newj,\newi) {};
			\ifthenelse{\NOT \i=0}{
				\tikzmath{
					integer \l;
					\l = \i - 1;}
				\draw (d\j\i) -- (d\j\l);
			}{}
		}
		\draw (d0\i) -- (d1\i);
	}
	\node[rwth-blue] at (5.7,2) {$1$};
	\node[rwth-blue] at (7.3,2) {$2$};
	
\foreach \i in {0,1,2,3}{
		\foreach \j in {0,1}{
			\tikzmath{
				integer \newj;
				\newj = \j+9;
				integer \newi;
				\newi = \i+4;}
			\node[smallVertex] (c\j\i) at (\newj,\newi) {};
			\ifthenelse{\NOT \i=0}{
				\tikzmath{
					integer \l;
					\l = \i - 1;}
				\draw (c\j\i) -- (c\j\l);
			}{}
		}
		\draw (c0\i) -- (c1\i);
	}
	\node[rwth-blue] at (8.7,4) {$1$};
	\node[rwth-blue] at (10.3,4) {$2$};
	\node[rwth-blue] at (10.3,5) {$3$};
	
\foreach \i in {0,1,2,3}{
		\foreach \j in {0,1}{
			\tikzmath{
				integer \newj;
				\newj = \j+9;
				integer \newi;
				\newi = \i-1;}
			\node[smallVertex] (d\j\i) at (\newj,\newi) {};
			\ifthenelse{\NOT \i=0}{
				\tikzmath{
					integer \l;
					\l = \i - 1;}
				\draw (d\j\i) -- (d\j\l);
			}{}
		}
		\draw (d0\i) -- (d1\i);
	}
	\node[rwth-blue] at (8.7,2) {$1$};
	\node[rwth-blue] at (10.3,2) {$2$};
	\node[rwth-blue] at (10.3,1) {$3$};
	
\foreach \i in {0,1,2}{
		\foreach \j in {0,1}{
			\tikzmath{
				integer \newj;
				\newj = \j+12;
				integer \newi;
				\newi = \i+6;}
			\node[smallVertex] (c\j\i) at (\newj,\newi) {};
			\ifthenelse{\NOT \i=0}{
				\tikzmath{
					integer \l;
					\l = \i - 1;}
				\draw (c\j\i) -- (c\j\l);
			}{}
		}
		\draw (c0\i) -- (c1\i);
	}
	\node[smallVertex] (c) at (12,5) {};
	\draw (c) -- (c00);
	\node[rwth-red] at (11.7,5) {$1$};
	\node[rwth-red] at (13.3,6) {$3$};
	
\foreach \j in {0,1,2}{
		\tikzmath{
		integer \newj;
		\newj = \j + 12;}
		\node[smallVertex] (e\j) at (\newj,4) {};
		\node[smallVertex] (f\j) at (\newj,2) {};
		\ifthenelse{\NOT \j=0}{
			\tikzmath{
				integer \l;
				\l = \j - 1;}
			\draw (e\j) -- (e\l);
			\draw (f\j) -- (f\l);
		}{}
	}
	\node[rwth-blue] at (12,4.3) {$1$};
	\node[rwth-blue] at (13,4.3) {$2$};
	\node[rwth-blue] at (14,4.3) {$3$};
	\node[rwth-blue] at (12,2.3) {$1$};
	\node[rwth-blue] at (13,2.3) {$2$};
	\node[rwth-blue] at (14,2.3) {$3$};
	
\foreach \i in {0,1,2}{
		\foreach \j in {0,1}{
			\tikzmath{
				integer \newj;
				\newj = \j+12;
				integer \newi;
				\newi = \i-2;}
			\node[smallVertex] (d\j\i) at (\newj,\newi) {};
			\ifthenelse{\NOT \i=0}{
				\tikzmath{
					integer \l;
					\l = \i - 1;}
				\draw (d\j\i) -- (d\j\l);
			}{}
		}
		\draw (d0\i) -- (d1\i);
	}
	\node[smallVertex] (d) at (12,1) {};
	\draw (d) -- (d02);
	\node[rwth-red] at (11.7,1) {$1$};
	\node[rwth-red] at (13.3,0) {$3$};
	
\foreach \i in {0,1,2}{
		\foreach \j in {0,1}{
			\tikzmath{
				integer \newj;
				\newj = \j+15;
				integer \newi;
				\newi = \i+6;}
			\node[smallVertex] (c\j\i) at (\newj,\newi) {};
			\ifthenelse{\NOT \i=0}{
				\tikzmath{
					integer \l;
					\l = \i - 1;}
				\draw (c\j\i) -- (c\j\l);
			}{}
		}
		\draw (c0\i) -- (c1\i);
	}
	\node[smallVertex] (c) at (15,5) {};
	\draw (c) -- (c00);
	\node[rwth-blue] at (14.7,5) {$1$};
	\node[rwth-blue] at (14.7,6) {$2$};
	\node[rwth-blue] at (16.3,6) {$3$};
	
\foreach \i in {0,1,2}{
		\foreach \j in {0,1}{
			\tikzmath{
				integer \newj;
				\newj = \j+15;
				integer \newi;
				\newi = \i-2;}
			\node[smallVertex] (d\j\i) at (\newj,\newi) {};
			\ifthenelse{\NOT \i=0}{
				\tikzmath{
					integer \l;
					\l = \i - 1;}
				\draw (d\j\i) -- (d\j\l);
			}{}
		}
		\draw (d0\i) -- (d1\i);
	}
	\node[smallVertex] (d) at (15,1) {};
	\draw (d) -- (d02);
	\node[rwth-blue] at (14.7,1) {$1$};
	\node[rwth-blue] at (14.7,0) {$2$};
	\node[rwth-blue] at (16.3,0) {$3$};
	
\foreach \i in {0,1,2}{
		\foreach \j in {0,1}{
			\tikzmath{
				integer \newj;
				\newj = \j+18;
				integer \newi;
				\newi = \i+6;}
			\node[smallVertex] (c\j\i) at (\newj,\newi) {};
			\ifthenelse{\NOT \i=0}{
				\tikzmath{
					integer \l;
					\l = \i - 1;}
				\draw (c\j\i) -- (c\j\l);
			}{}
		}
		\draw (c0\i) -- (c1\i);
	}
	\node[rwth-blue] at (17.7,6) {$2$};
	\node[rwth-blue] at (19.3,6) {$3$};
	
\foreach \j in {0,1,2}{
		\tikzmath{
			integer \newj;
			\newj = \j + 18;}
		\node[smallVertex] (e\j) at (\newj,4) {};
		\node[smallVertex] (f\j) at (\newj,2) {};
		\ifthenelse{\NOT \j=0}{
			\tikzmath{
				integer \l;
				\l = \j - 1;}
			\draw (e\j) -- (e\l);
			\draw (f\j) -- (f\l);
		}{}
	}
	\node[rwth-blue] at (18,4.3) {$1$};
	\node[rwth-blue] at (19,4.3) {$2$};
	\node[rwth-blue] at (20,4.3) {$3$};
	\node[rwth-blue] at (18,2.3) {$1$};
	\node[rwth-blue] at (19,2.3) {$2$};
	\node[rwth-blue] at (20,2.3) {$3$};
	
\foreach \i in {0,1,2}{
		\foreach \j in {0,1}{
			\tikzmath{
				integer \newj;
				\newj = \j+18;
				integer \newi;
				\newi = \i-2;}
			\node[smallVertex] (d\j\i) at (\newj,\newi) {};
			\ifthenelse{\NOT \i=0}{
				\tikzmath{
					integer \l;
					\l = \i - 1;}
				\draw (d\j\i) -- (d\j\l);
			}{}
		}
		\draw (d0\i) -- (d1\i);
	}
	\node[rwth-blue] at (17.7,0) {$2$};
	\node[rwth-blue] at (19.3,0) {$3$};
	
\foreach \i in {0,1,2}{
		\foreach \j in {0,1}{
			\tikzmath{
				integer \newj;
				\newj = \j+21;
				integer \newi;
				\newi = \i+6;}
			\node[smallVertex] (c\j\i) at (\newj,\newi) {};
			\ifthenelse{\NOT \i=0}{
				\tikzmath{
					integer \l;
					\l = \i - 1;}
				\draw (c\j\i) -- (c\j\l);
			}{}
		}
		\draw (c0\i) -- (c1\i);
	}
	\node[rwth-blue] at (20.7,7) {$1$};
	\node[rwth-blue] at (20.7,6) {$2$};
	\node[rwth-blue] at (22.3,6) {$3$};
	
\foreach \i in {0,1,2}{
		\foreach \j in {0,1}{
			\tikzmath{
				integer \newj;
				\newj = \j+21;
				integer \newi;
				\newi = \i-2;}
			\node[smallVertex] (d\j\i) at (\newj,\newi) {};
			\ifthenelse{\NOT \i=0}{
				\tikzmath{
					integer \l;
					\l = \i - 1;}
				\draw (d\j\i) -- (d\j\l);
			}{}
		}
		\draw (d0\i) -- (d1\i);
	}
	\node[rwth-blue] at (20.7,-1) {$1$};
	\node[rwth-blue] at (20.7,0) {$2$};
	\node[rwth-blue] at (22.3,0) {$3$};
	
\foreach \i in {0,1}{
		\foreach \j in {0,1}{
			\tikzmath{
				integer \newj;
				\newj = \j+25;
				integer \newi;
				\newi = \i+6;}
			\node[smallVertex] (c\j\i) at (\newj,\newi) {};
			\ifthenelse{\NOT \i=0}{
				\tikzmath{
					integer \l;
					\l = \i - 1;}
				\draw (c\j\i) -- (c\j\l);
			}{}
		}
		\draw (c0\i) -- (c1\i);
	}
	\node[rwth-blue] at (24.7,6) {$1$};
	\node[rwth-blue] at (26.3,6) {$2$};
	\node[rwth-blue] at (26.3,7) {$3$};
	
\foreach \i in {0,1}{
		\foreach \j in {0,1}{
			\tikzmath{
				integer \newj;
				\newj = \j+25;
				integer \newi;
				\newi = \i-1;}
			\node[smallVertex] (d\j\i) at (\newj,\newi) {};
			\ifthenelse{\NOT \i=0}{
				\tikzmath{
					integer \l;
					\l = \i - 1;}
				\draw (d\j\i) -- (d\j\l);
			}{}
		}
		\draw (d0\i) -- (d1\i);
	}
	\node[rwth-blue] at (24.7,0) {$1$};
	\node[rwth-blue] at (26.3,0) {$2$};
	\node[rwth-blue] at (26.3,-1) {$3$};
	
\foreach \j in {0,1,2}{
		\tikzmath{
			integer \newj;
			\newj = \j + 28;}
		\node[smallVertex] (c\j) at (\newj,7) {};
		\node[smallVertex] (d\j) at (\newj,-1) {};
		\node[smallVertex] (e\j) at (\newj,4) {};
		\node[smallVertex] (f\j) at (\newj,2) {};
		\ifthenelse{\NOT \j=0}{
			\tikzmath{
				integer \l;
				\l = \j - 1;}
			\draw (c\j) -- (c\l);
			\draw (d\j) -- (d\l);
			\draw (e\j) -- (e\l);
			\draw (f\j) -- (f\l);
		}{}
	}
	\node[rwth-blue] at (28,4.3) {$1$};
	\node[rwth-blue] at (29,4.3) {$2$};
	\node[rwth-blue] at (30,4.3) {$3$};
	\node[rwth-blue] at (28,2.3) {$1$};
	\node[rwth-blue] at (29,2.3) {$2$};
	\node[rwth-blue] at (30,2.3) {$3$};
	\node[rwth-red] at (28,7.3) {$1$};
	\node[rwth-red] at (30,7.3) {$3$};
	\node[rwth-red] at (28,-0.7) {$1$};
	\node[rwth-red] at (30,-0.7) {$3$};
	
\foreach \j in {0,1,2}{
		\tikzmath{
			integer \newj;
			\newj = \j + 32;}
		\node[smallVertex] (e\j) at (\newj,7) {};
		\node[smallVertex] (f\j) at (\newj,-1) {};
		\ifthenelse{\NOT \j=0}{
			\tikzmath{
				integer \l;
				\l = \j - 1;}
			\draw (e\j) -- (e\l);
			\draw (f\j) -- (f\l);
		}{}
	}
	\node[rwth-blue] at (32,7.3) {$1$};
	\node[rwth-blue] at (33,7.3) {$2$};
	\node[rwth-blue] at (34,7.3) {$3$};
	\node[rwth-blue] at (32,-0.7) {$1$};
	\node[rwth-blue] at (33,-0.7) {$2$};
	\node[rwth-blue] at (34,-0.7) {$3$};
	
\draw[dashed, thick, rwth-orange] (1.4,3) -- (2.6,3) (7.4,5) -- (8.6,5) (7.4,1) -- (8.6,1) (13.4,7) -- (14.6,7) (13.4,-1) -- (14.6,-1) (19.4,7) -- (20.6,7) (19.4,-1) -- (20.6,-1) (30.4,7) -- (31.6,7) (30.4,-1) -- (31.6,-1);
	\draw[thick, rwth-orange] (5.6,1) -- (4.5,3) -- (5.6,5) (11.6,4) -- (10.5,5) -- (11.6,6) (11.6,0) -- (10.5,1) -- (11.6,2) (17.6,4) -- (16.5,6) -- (17.6,7) (17.6,-1) -- (16.5,0) -- (17.6,2) (27.6,4) -- (26.5,6) -- (27.6,7) (27.6,-1) -- (26.5,0) -- (27.6,2);
	\draw[thick, dotted, rwth-orange] (22.4,7) -- (24.6,7) (22.4,-1) -- (24.6,-1);
  \end{tikzpicture}
  \caption{A guarded $3$-construction tree of elimination depth $7$ for the grid $\grid{2}{7}$ with one labelled vertex. Edges entering elimination nodes are dashed. At every labelled graph, those labels that may be removed are marked blue, those that may not be removed are marked red. The omitted part of the construction tree follows the same pattern.}
  \label{fig:guard-elim-example}
\end{figure}

\begin{restatable}{lemma}{gcHomcap}\label{lem:gchomcap}
  Let $F \in \GEkq$. Then for each $m \geq 0$ there is a formula $\varphi_m$ such that for appropriately labelled graphs $G$ it holds that $\hom(F, G) = m$ iff $G \models \varphi_m$.
\end{restatable}

  The proof of the converse---showing that there exists for each $\psi \in \lGC^k_q$ an $\qg{F} \in \R\GEkq$ modelling $\psi$---also goes through nearly unchanged.

\begin{restatable}{lemma}{qgFromGckq}\label{lem:qg_from_gckq}
  Let $\varphi \in \lGC^k_q$. Then there is an $\qg{F} \in \R\GEkq$ modelling $\varphi$ for graphs of size $n$.
\end{restatable}

The analogues of these two lemmas already sufficed to prove \cref{thm:ckq_equivalence}. Here, however, we still need to be mindful of any remaining labels. Concretely, \cref{lem:gchomcap} and \cref{lem:qg_from_gckq} imply the following for $\lGC$ sentences.

\begin{corollary}\label{cor:gc-thm-with-labels}
  Let $G, v$ and $H, w$ be graphs together with a single labelled vertex. Then the following are equivalent.
  \begin{enumerate}
    \item For all $\psi(x) \in \lGC^k_q$, it holds $G, v \models \psi(x) \iff H, w \models \psi(x)$.
    \item $\hom(F, G) = \hom(F, H)$ for all $F \in \GEkq$.
  \end{enumerate}
\end{corollary}

While this is already a nice result, ideally we would like to make a statement about general, unlabelled, graphs. Fortunately, simply removing all labels from $F \in \GEkq$ turns out to induce the equivalence relation described in \cref{def:gcequiv}.

Let us denote by $\GEkqLL$ the class of graphs in $\GEkq$ with all labels removed. Then we can state the following theorem, characterising $\lGC^k_q$-equivalence in terms of homomorphism indistinguishability.
The details can be found in \Cref{app:gc-homind}.

\guardedEkqLogic*

\subsection{Separating $\Ekq$ from $\mathcal{TW}_{k-1} \cap \mathcal{TD}_q$ Semantically}
By \cref{thm:Ekq_tw-td}, the graph class $\Ekq$ is a proper subclass of $\mathcal{TW}_{k-1} \cap \mathcal{TD}_q$. Despite that, it could well be that the homomorphism indistinguishability relations of the two graph classes (and via \cref{thm:ckq_equivalence} also $\mathsf{C}^k_q$-equivalence) coincide, i.e.\@ $G \equiv_{\Ekq} H$ if and only if $G \equiv_{\mathcal{TW}_{k-1} \cap \mathcal{TD}_q} H$ for all graphs $G$ and $H$. It turns out that this is not the case.

In general, establishing that the homomorphism indistinguishability relations $\equiv_{\mathcal{F}_1}$ and $\equiv_{\mathcal{F}_2}$ of two graph classes $\mathcal{F}_1 \neq \mathcal{F}_2$ are distinct is a notoriously hard task. 
Pivotal tools for accomplishing this were introduced by Roberson in \cite{roberson_oddomorphisms_2022}. He defines the \emph{homomorphism distinguishing closure} $\cl(\mathcal{F})$ of a graph class $\mathcal{F}$ as the graph class
\[
\cl(\mathcal{F}) \coloneqq \{ F \text{ graph} \mid \forall G, H.\ G \equiv_{\mathcal{F}} H \implies \hom(F, G) = \hom(F, H)\}.
\]
A graph class $\mathcal{F}$ is \emph{homomorphism distinguishing closed} if $\mathcal{F} = \cl(\mathcal{F})$. In this case, for every $F \not\in \mathcal{F}$ there exist two graphs $G$ and $H$ homomorphism indistinguishable over $\mathcal{F}$ and satisfying that $\hom(F, G) \neq \hom(F, H)$. 
Therefore, homomorphism distinguishing closed graph classes may be thought of as maximal in terms of homomorphism indistinguishability.

Roberson conjectures that \emph{every graph class which is closed under taking minors and disjoint unions is homomorphism distinguishing closed}. A confirmation of this conjecture would aid separating homomorphism indistinguishability relations and in turn all equivalence relations between graphs which have such characterisations, cf.\@ \cite{roberson_lasserre_2023}. 
In particular, it would readily imply that $\equiv_{\Ekq}$ and $\equiv_{\mathcal{TW}_{k-1} \cap \mathcal{TD}_q}$ are distinct, cf.\@ \cref{cor:minor-closed}.
Unfortunately, the conjecture's assertion is only known to be true for the class of planar graphs \cite{roberson_oddomorphisms_2022}, $\mathcal{TW}_k$ \cite{neuen_homomorphism-distinguishing_2023} and graph classes arising from finite graph classes \cite{seppelt_logical_2023}.
Towards separating $\equiv_{\Ekq}$ and $\equiv_{\mathcal{TW}_{k-1} \cap \mathcal{TD}_q}$, we first add to this list by proving the following:

\tdclosed*

The proof of \cref{thm:td-closed} follows the proof in \cite{neuen_homomorphism-distinguishing_2023} of the assertion that the class $\mathcal{TW}_{k}$ is homomorphism distinguishing closed for all $k \geq 0$.
Central to it is a construction of highly similar graphs from \cite{roberson_oddomorphisms_2022} which is reminiscent of the CFI-construction \cite{cai_optimal_1992}. 
With these ingredients, it suffices to prove that Duplicator wins the model comparison game characterising $\mathsf{C}_q$-equivalence on these CFI-like graphs constructed over a graph of high treedepth. To that end, we build a Duplicator strategy from a Robber strategy for the game from \cref{def:Ekq-cops}.
The connection between model comparison and node searching games via CFI-constructions is well-known \cite{hella_logical_1996,dawar_power_2007}.

Crucial for the aforementioned argument is that Robber wins the \emph{non-monotone} node searching game characterising bounded treedepth. Indeed, it cannot be assumed that Cops plays monotonously since he must shadow Spoiler's moves. Since we are unable to establish a non-monotone node searching game characterising $\Ekq$, we cannot conclude along the same lines that $\Ekq$ is homomorphism distinguishing closed. Nevertheless, we separate $\equiv_{\Ekq}$ and $\equiv_{\mathcal{TW}_{k-1} \cap \mathcal{TD}_q}$.
The details are deferred to \cref{app:td-closed}.
\ekqtwtdsemantics*

Proving that $\Ekq$ is characterised by Robber winning the (non-monotone) game $\CR_q^k$, i.e.\@ the non-monotone version of \cref{lem:Ekq-cops}, would immediately imply that $\Ekq$ is homomorphism distinguishing closed. 
\section{Outlook}
We studied the expressive power of the counting logic fragment $\mathsf{C}^k_q$ with tools from homomorphism indistinguishability. After giving an elementary and uniform proof of theorems from \cite{dvorak_recognizing_2010,grohe_counting_2020,dawar_lovasz-type_2021}, we showed that the graph class $\Ekq$, whose homomorphism indistinguishability relation characterises $\mathsf{C}^k_q$-equivalence, is a proper subclass of $\mathcal{TW}_{k-1} \cap \mathcal{TD}_q$. Finally, we showed that homomorphism indistinguishability over $\Ekq$ is not the same as homomorphism indistinguishability over $\mathcal{TW}_{k-1} \cap \mathcal{TD}_q$.

The main problem remaining open is to tighten \cref{thm:Ekq_tw-td-semantics} by proving that the graph class $\Ekq$ is homomorphism distinguishing closed, as predicted by Roberson's conjecture. Our contribution in this direction is a reduction to a purely graph theoretic problem: Proving that the class $\Ekq$ is characterised by a \emph{non-monotone} cops-and-robber game, cf.\@ \cref{lem:Ekq-cops}, would be sufficient to yield this claim. Exploring whether intertwining node searching and model comparison games can help to verify Roberson's conjecture in other cases seems a tempting direction for future research.

With slight reformulations, our results might yield insights into the ability of the Weisfeiler--Leman algorithm to determine subgraph counts after a fixed number of rounds \cite{rattan_weisfeiler-leman_2023,neuen_homomorphism-distinguishing_2023}.

\newpage
\appendix

\section{Material Omitted in \Cref{sec:graph-dec}}
\label{app:graph-dec}
\newcommand{\numberOfRoundsK}{\ensuremath{\frac{(k-1) (\ell  - k + 3)}{4}}}
\newcommand{\sizeLargeComponent}{\ensuremath{\frac{\ell h - h - 2}{2}}}
\newcommand{\sizeNonGood}{\ensuremath{\frac{(h-1)(h+2)}{2}}}
\newcommand{\boundL}{\ensuremath{\lceil \frac{q}{k-1}\rceil (k+4)}}

\subsection{Equivalence of Definitions (Proof of \cref{thm:Ekq_equiv})}
\label{app:equiv-classes}

\ekqEquiv*

\begin{proof}
	We first show $(1)\Leftrightarrow (2)$.
	
	Let $G$ be \elimOrd{k}{q} witnessed by a $k$-construction tree $(T,\lambda,r)$ of \elimDepth\ $\leq q$.
	For every node $t\in V(T)$ and the corresponding $k$-labelled graph $\lambda(t)$, we write $\beta(t)$ for the set of labelled vertices of $\lambda(t)$.
	We show that $(T,\beta,r)$ is a tree decomposition of $G$ of width $\leq k-1$ and depth $\leq q$.
	We observe that the number of vertices that are in any bag of a path from the root to some leaf equal the number of times a label was removed on the same path and therefore is bounded by the \elimDepth.
	Thus one only has to show that this is indeed a tree decomposition.
	For every edge $uv\in E(G)$, there has to be some leaf $\ell\in V(T)$, such that the edge is already present in $\lambda(\ell)$. Hence, $u,v \in \beta(\ell)$.
	Finally, we have that $\beta^{-1}(v)$ is the subset of nodes $t\in V(T)$ such that $v\in V(\lambda(t))$ and $v$ is labelled in $\lambda(t)$.
	This sub-tree is obviously connected, as labels can only be introduced at leaf nodes and only labelled vertices are identified at join nodes.
	
	Conversely, let $(T,r,\beta)$ be a rooted tree decomposition of $G$ of width $\leq k-1$ and depth $\leq q$.
	A rooted tree decomposition is \emph{nice} if every node that is not a leaf is either an introduce node, a forget node or a join node.
	We call a node $t$ 
	\begin{itemize}
		\item \emph{introduce node} if it has exactly one child $s$ such that there exists a vertex $v\in V(G)$ with $\beta(s)=\beta(s)\cup \{v\} $,
		\item \emph{forget node} if it has exactly one child $s$ such that there exists a vertex $v\in V(G)$ with $\beta(s)=\beta(t)\cup \{v\}$, and 
		\item \emph{join node} if it has exactly two children $s_1,s_2$ with $\beta(t)=\beta(s_1)=\beta(s_2)$.
	\end{itemize}
	Additionally the bag of the root node and of all leaf nodes are empty.
	By \cite{bodlaender_partial_1998}, if there is a tree decomposition then there also is a nice tree decomposition.
	We observe that the technique to make a tree decomposition nice preserves the depth of the decomposition.
	Thus w.l.o.g. $(T,r,\beta)$ is nice.
	We observe that the join node of a nice tree decomposition is quite similar to the product node of a construction tree, as is the forget node to the eliminate node.
	But we need to get rid of the introduce nodes.
	Therefore we construct a new rooted tree decomposition $(T',r,\beta')$, where at every introduce node $t$ we append a new leaf $\ell_t$, with bag $\beta'(\ell_t)=\beta(t)$.
	Furthermore we set $\beta'(t)=\beta(t)$, for every $t\in V(T)$.
	For every $t\in V(T')$, we now set $\lambda(t):=G[\gamma(t)]$, that is the subgraph induced by all vertices that are in a bag of the sub-tree below $t$.
	We set $\beta(t)$ to be the labelled vertices of $\lambda(t)$.
	It remains to define a colouring function $c\colon V(G)\rightarrow [k]$, such that for all $t\in V(T')$, it holds that $c|_{\beta(t)}$ is injective.
	We define this colouring via traversing the tree $T'$ from the root to the leafs.
	Whenever $t$ is a forget node with child $s$, such that $\beta(s)\setminus\beta(t)=\{v\}$, we set $c(v)$ to be the smallest value that is not used for any vertex in $\beta(t)$.
	Then $T'$ together with the graphs $H_t$, that are labelled via the function $c|_{\beta(t)}^{-1}$ is the desired elimination order.
	
	Next we show $(2)\Leftrightarrow (3)$.
	We use the same construction as in the proof of \cite[Theorem~19]{abramsky_relating_2021}, where the authors gave a proof that a tree decomposition of width $\leq k-1$ exists if and only if a $k$-pebble forest cover exists.
	We recall their construction and prove that it preserves depth.
	
	Let $(T,\beta)$ be a tree decomposition of $G$, of width $\leq k-1$ and depth $\leq q$, and $r\in V(T)$ such that $\dep(T,\beta) = \dep(T,r,\beta)$.
	Again w.l.o.g. $(T,r,\beta)$ is nice.
	We define the function $\tau\colon V(G) \rightarrow V(T)$ to map every vertex of $G$ to the unique node of $t$ that has smallest distance to the root $r$.
	As $(T,r,\beta)$ is nice, this function is injective.
	The forest cover $(F,\vec{r}')$ is induced by the partial order on the image of $\tau$ with respect to $(T,r)$.
	The pebbling function $p\colon V(G) \rightarrow [k]$ is defined inductively.
	Assume $p$ is defined for all $v'\prec v$.
	Then $p(v)=\min [k]\setminus \{ p(v')\mid v'\in \beta(\tau(v))\setminus\{v\}\}$.
	We observe that the depth of $(F,\vec{r}')$ is the length of the longest chain in the image of $\tau$ with respect to $(T,r)$.
	Additionally we observe that for all $v\in V(G)$ and $t\in V(T)$ from $v\in\beta(t)$ it follows that $\tau(v)\preceq t$.
	Thus the depth of $(F,\vec{r}')$ equals the maximum number of vertices from the root $r$ to any leaf of $T$ and thus the depth of $(T,r,\beta)$.
	
	Now let $(F, \vec{r},p)$ be a $k$-pebble forest cover of depth $\leq q$ of $G$.
	We construct a rooted tree $(T,r')$ by introducing a new root $r'$ that connects to all nodes in $\vec{r}$.
	We define $\beta(r')=\emptyset$ and $\beta(t)\coloneqq \{u\preceq t\mid \text{ for all } w\in V(G), u\prec w\preceq t \Rightarrow p(u)\neq p(w) \}$, for every $t\in V(F)$.
	We have that $\dep(T,r,\beta)= \max_{v \in V(T)} \left| \bigcup_{t\in P_v} \beta(t) \right| = \max_{t\in V(F)}|\{s\in V(F)\mid s\preceq t \}|=\dep(F,\vec{r},p)$.
\end{proof}

\subsection{Monotone Cop Strategies and Decompositions (Proof of \cref{lem:Ekq-cops})}
\label{app:cr-equiv}

\crGame*

\begin{proof}
	Let $G$ be a graph. Let $1\leq k\leq q$.
	
	Assume $G\in\Ekq$, thus there is a $k$-pebble forest cover $(F,\vec{r},p)$ of depth at most $q$.
	For any set $Y\subseteq V(G)$ and any vertex $x\in V(G)$, we write $Y\preceq x$ and $x\preceq Y$ to denote that $y\preceq x$ and respectively $x\preceq y$ with respect to $(F,\vec{r})$, for all $y\in Y$, and $\max_{\preceq}Y$ to denote the maximal element with respect to $(F,\vec{r})$.
	We extend the pebbling function $p$ to a pebbling function $p'$ on $V(G')$ by an arbitrary bijection between $[k]$ and the vertices of $K$.
	From this we construct a winning strategy for the cops in the game $\monCR_q^k(G)$ such that for every position $(X,y)$ of the game the following two conditions on $X'\coloneqq X\cap V(G)$ hold:
	\begin{enumerate}
		\item $X'\preceq y$ and
		\item $X'=\{u\preceq \max_{\preceq} X'\mid \text{ for all } w\in V(G) \text{ with }u\prec w\preceq \max_{\preceq} X', p(u)\neq p(w) \}$.
	\end{enumerate}
	The strategy is then defined as
	\begin{equation*}
	\mathcal{S}(X,y)\coloneqq \left(X\setminus \{x_b\}\right)\cup \{b\},
	\end{equation*}
	where $b$ is the minimal vertex that fulfils $X'\prec b\preceq y$ and $x_b\in X$ the unique element with $p'(x_b)=p(b)$.
	Then $(2)$ holds by observing that $b$ either is some root vertex if there are no cops on $G$ and otherwise $b$ is a child of $\max_{\preceq} X'$.
	To see that $(1)$ holds let us assume that there is some $u$ reachable from $y$ in $G'\setminus \left(X\setminus\{x_b\}\right)$ such that $b\not\preceq u$.
	Since $(F,\vec{r})$ is a forest cover of $G$ we know that every path from $y$ to $u$ meets some vertex that comes before $b$ with respect to $(F,\vec{r})$ and thus the $X\setminus\{x_b\}$ avoiding path $P$ from $y$ to $u$ contains some edge $vw\in E(G)$ such that $v\prec \max_{\preceq} X' \prec b \preceq w$.
	We observe that all vertices of $X'$ are comparable to $v$, as by $(2)$ all vertices in $X'$ are pairwise comparable and thus on the unique path from some root in $\vec{r}$ to $\max_{\preceq} X'$.
	By $(2)$ we get that there is some vertex $x\in X$ with $p(x)=p(v)$ and also $v\prec x$.
	Thus $v\prec x\prec b\preceq w$ yields a contradiction to $(F,\vec{r},p)$ being a $k$-pebble forest cover.
	
	All in all the robber is forced down the tree in $(F,\vec{r})$, that corresponds to the connected component that the robber chose. As $(F,\vec{r})$ has depth $q$, the robber is caught after at most $q$ rounds, thus $\mathcal{S}$ is a winning strategy for the cops in the game $\monCR_q^k(G)$
	
	For the other direction let $\mathcal{S}$ be a winning-strategy for Cops in the game $\monCR_q^k(G)$.
	W.l.o.g. it holds that, for all $X\in\binom{V(G')}{k}$, all $y\in V(G)\setminus X$ and all $y'\in \gamma_y^X$, we have $\mathcal{S}(X,y)=\mathcal{S}(X,y')$.
	We write $\mathcal{S}(X,\gamma_y^X):=\mathcal{S}(X,y)$.
	Additionally w.l.o.g. we can assume that the cops always place the next cop within the robbers escape space, thus $\mathcal{S}(X,V(C))\subseteq X\cup V(C)$, for all $X\in\binom{V(G')}{k}$ and $C$ connected component of $G\setminus X$.
	We construct a tree decomposition of width $\leq k-1$ and depth $\leq q$ from this strategy, thus by \cref{thm:Ekq_equiv} it holds that $G\in\Ekq$.
	
	We inductively define the ``strategy tree'' describing the strategy $\mathcal{S}$, that is a rooted tree $(T,r)$, a function $\lambda\colon V(T)\setminus\{r\}\to \binom{V(G')}{k}\times 2^{V(G)}$ and a function $\beta\colon V(T)\rightarrow \binom{V(G)}{\leq k}$ by
	\begin{itemize}
		\item $\beta(r)\coloneqq \emptyset$,
		\item let $C_1,\ldots,C_m$ be the connected components of $G$, then add children $t_1,\ldots,t_m$ to $r$ with $\lambda(t_i)=(V(K),V(C_i))$ and $\beta(t_i)=\mathcal{S}(V(K),V(C_i))\setminus V(K)$ and
		\item for a node $s\in V(T)\setminus\{r\}$ with $\lambda(s)=(X,Y)$, let $X':=\mathcal{S}(X,Y)$ and let $C_1,\ldots,C_m$ be the connected components of $Y\setminus X'$, then add children $t_1,\ldots,t_m$ to $s$ with $\lambda(t_i)=(X',V(C_i))$ and $\beta(t_i)=\mathcal{S}(X',V(C_i))\setminus V(K)$.
	\end{itemize}
	Then $(T,r,\beta)$ is the desired tree decomposition. 

	\begin{claim}
		$(T,r,\beta)$ is a rooted tree decomposition of $G$.
	\end{claim}
	
	\begin{claimproof}
		We first prove that the pre-image of every vertex $v\in V(G)$ is non-empty and connected.
		$\beta^{-1}(v)$ is non-empty as the robber is eventually caught and thus he is also caught if he never leafs vertex $v$.
		Now suppose for contradiction that there are distinct non-adjacent nodes $t_1,t_2\in\beta^{-1}(v)$ such that, for all $s$ on the path from $t_1$ to $t_2$, $s\not\in\beta^{-1}(v)$.
		
		Let us first consider the case that neither $t_1\prec t_2$ nor $t_2\prec t_1$ with respect to $(T,r)$.
		Let $s\in V(T)$ be the maximal node such that $s\prec \{t_1,t_2\}$ and $(X,Y)\coloneqq\lambda(s)$.
		Then $v\not\in\beta(s)=\mathcal{S}(\lambda(s))\setminus V(K)$.
		Thus there is one child $s'$ of $s$ that corresponds to the component $\gamma_v^{\mathcal{S}(\lambda(s))}$.
		As we assumed that Cops always places the cops into the robber component we get that both $t_1$ and $t_2$ are descendants of $s'$, which contradicts the maximality of $s$.
		
		Thus assume $t_1 \prec t_2$ ($t_2\prec t_1$ is symmetric).
		Let $s$ be the child of $t_1$ with $s\prec t_2$ and let $(X,Y)\coloneqq\lambda(s)$.
		We get that $v\in\beta(s)\cup Y$, again by the assumption that Cops always places the cops into the robber component.
		This contradicts the monotonicity of $\mathcal{S}$, as $v\notin \beta(s)$ by the choice of $t_1$ and $t_2$.
		
		It remains to show that for every edge $uv\in E(G)$ there is some node $t\in V(T)$ such that $u,v\in\beta(t)$.
		Assume for contradiction that $\beta^{-1}(u)$ and $\beta^{-1}(v)$ are disjoint, for some edge $uv\in E(G)$.
		Let $t_1\in \beta^{-1}(u)$ and $t_2\in \beta^{-1}(v)$ such that $s\notin\beta^{-1}(u)\cup\beta^{-1}(v)$, for all $s$ on the path from $t_1$ to $t_2$.
		We observe that $t_1$ and $t_2$ are comparable, as $\gamma_{v}^X=\gamma_u^X$, for all $X\in\binom{V(G')\setminus\{u,v\}}{k}$.
		W.l.o.g. assume $t_1\prec t_2$.
		Let $s$ be the child of $t_1$ such that $s\preceq t_2$ and let $(X,Y)\coloneqq\lambda(s)$.
		As $v\in Y$ and $u\notin\mathcal{S}(\lambda(s))$, the robber can escape to $v$, which contradicts monotonicity.
	\end{claimproof}
	
	The width of $(T,r,\beta)$ is at most $k-1$.
	As a last step we need to verify the size of $\bigcup_{t\preceq s}\beta(t)$, for all vertices $s\in V(T)$.
	As the cop player can move at most one cop per round, we now that for all $t\in V(T)$ and all children $s$ of $t$, there is at most one vertex in $\beta(s)$ that is not in $\beta(t)$.
	As the bag of the root is empty and every edge in the tree corresponds to a Cops move in the game, we know that from the root to any node there can not be more than $q$ introduced vertices.
	Thus $|\bigcup_{t\preceq s}\beta(t)|\leq q$, for all $s\in V(T)$.
\end{proof}

\subsection{The Grid is Narrow and Shallow but not at the same Time (Proofs of \cref{lem:lower-bound-rounds} and \cref{thm:Ekq_tw-td})}
\label{app:rounds}

In order to prove \cref{lem:lower-bound-rounds}, we first make some structural observations on the grid~\grid{h}{\ell} with $h>1$ rows and $\ell$ columns and its separators of size $\leq h+1$.
We write $V(\grid{h}{\ell})=\{(1,1),\ldots,(1,\ell),(2,1),\ldots,(h,\ell)\}$ and $(i,j)(i',j')\in E(\grid{h}{\ell})$ if $i'=i\pm 1$ or $j'=j\pm 1$ but not both.
For some $X\subseteq V(\grid{h}{\ell})$, we call a connected component of $\grid{h}{\ell}\setminus X$ \emph{good} if it contains at least one full column of \grid{h}{\ell}.

\begin{observation} \label{obs:good-component}
	For $1 < h < \ell - 1$ and $X\in \binom{V(\grid{h}{\ell})}{\leq h+1}$, there exists a good component in $\grid{h}{\ell} \setminus X$.
\end{observation}
\begin{proof}
	Since there are at most $h+1$ vertices in $X$ and the grid has at least $h+2$ columns, there exists a column which does not contain any vertices from $X$. This column is contained by a good component of $\grid{h}{\ell} \setminus X$.
\end{proof}
The next lemma shows that $\grid{h}{\ell}\setminus X$ can never contain more than two good components if $X\in\binom{V(\grid{h}{\ell}}{\leq h+1}$.
And if it does contain two good components, there can only be a single vertex that is in any other component.

\begin{lemma}
	\label{lem:third-comp-single}
	Let $h,\ell>1$ and let $X\in \binom{V(\grid{h}{\ell})}{\leq h+1}$. If there are two distinct good components in $\grid{h}{\ell}\setminus X$ then there is at most one additional component and this component has size 1.
\end{lemma}

\begin{proof}
	Let $C_1,C_2$ be the two good components.
	Since $X$ separates $C_1$ and $C_2$, $X$ must contain at least one vertex from every row. Thus, at most one row contains two vertices from $X$.
	Let $C_3$ be the set of vertices that are not contained in $X\cup C_1\cup C_2$.
	Assume $C_3$ contains a vertex from row $i\in[h]$.
	Then row $i$ contains vertices from $C_1$, $C_2$, and $C_3$. 
	Since these are all distinct components of $\grid{h}{\ell}\setminus X$, row $i$ also contains at least two, thus exactly two, vertices of $X$.
	
	Hence, $C_3$ intersects at most one row, say row~$r$.
	As $C_3$ cannot contain a vertex of any other row, all neighbours of $C_3$ that are not in row $r$ need to be in $X$.
	As $h>1$, at least one of $r-1$ or $r+1$ is in $[h]$.
	Thus there is a row in $[h]\setminus \{r\}$ where at least $|C_3|$ vertices are contained in $X$.
	Since all rows besides $r$ have exactly one vertex in $X$, we conclude that $|C_3|\leq 1$.
\end{proof}

We use \cref{lem:third-comp-single} to show that Robber can find some large component and that the cops can never remove more than two vertices from this component as long as his component is good.
Next we show that a large component is always good.

\begin{lemma}
	\label{lem:size-non-good}
	Let $1<h<\ell-2$ and let $X\in \binom{V(\grid{h}{\ell})}{\leq h+1}$. A connected component $C$ of $\grid{h}{\ell}\setminus X$
	that contains more than $\sizeNonGood$ vertices is good.
\end{lemma}

\begin{proof}
	Let $C$ be a connected component of $\grid{h}{\ell}\setminus X$ which is not good. We prove that $C$ contains at most $\sizeNonGood$ many vertices.
	
	By \cref{obs:good-component}, there exists a good component $C'$ in $G \setminus X$ containing some column $j \in [\ell]$. Without loss of generality, all columns intersecting $C$ lie left of $j$, i.e.\@ have index strictly smaller than $j$.
	
	\begin{claim} \label{cl:size-non-good}
		For $1\leq i<h$,
		$C$ contains at most $i+1$ columns with at least $h-i$ vertices.
	\end{claim}
	\begin{claimproof}
		Assume $C$ would contain at least $i+2$ columns with at least $h-i$ vertices.
		$X$ contains at least one vertex of each of these columns, as $C$ is not good.
		Let $j'$ be the largest index of a column intersecting $C$ and let $I \subseteq [h]$ be the set of rows that intersect $C$ in column $j'$, that is $(i',j')\in C$, for all $i'\in I$.
		We define $|I|$ many vertex disjoint $C$-$C'$-paths by $(i',j'),(i',j'+1),\ldots,(i',j)$ for $i' \in I$.
		By Menger's Theorem, $X$ needs to contain at least one vertex of each of these paths and all those vertices are in some column that is strictly larger than $j'$.
		All in all, $X$ contains at least $i+2+|I|\geq i+2+h-i=h+2$ vertices, which contradicts the choice of $X$.
	\end{claimproof}
	
	For $1 \leq i \leq h$, write $m_i$ for the number of columns in $C$ which contain at least $i$ vertices. Then $m_i - m_{i+1}$ is the number of columns in $C$ which contain exactly $i$ vertices. 
	By \cref{cl:size-non-good}, $m_i \leq h +1 - i$ for all $1 \leq i \leq h-1$.
	Hence,
	\[
		|C| = \sum_{i = 1}^{h-1} (m_i - m_{i+1}) i = \sum_{i=1}^{h-1} m_i \leq (h-1)(h+1) - \sum_{i=1}^{h-1} i = \sizeNonGood,
	\]
	as desired.
\end{proof}

We observe that this bound is tight as one can realise a component $C$ that is not good and contains exactly $i+1$ columns with at least $h-i$ vertices, for all $1\leq i<h$, by $X=\{(1,1),(1,2),(2,3),\ldots,(h,h+1)\}$.

\lowerBoundRounds*

\begin{proof}The strategy of Robber is to always move to the largest component.
	\begin{claim}
		After $q'< \frac{\ell h - h^2 + 2h}{4}$ rounds the size of the robber component is at least $\sizeLargeComponent-2(q'-h) > \sizeNonGood$.
	\end{claim}
	\begin{claimproof}
		Since $h < \ell-2$, for every cop position $X\in \binom{V(\grid{h}{\ell})}{\leq h+1}$, there are at least two columns where no cop is positioned.
		In order to catch the robber, the cops have to move to a position such that these two columns lie in distinct good components of  $\grid{h}{\ell}\setminus X$.
		Since such an $X$ contains a vertex of every row, this takes the cops at least $h$ rounds.
		
		By the pigeonhole principle and \cref{lem:third-comp-single}, there is some component $C$ of size at least $\sizeLargeComponent$ and it can be reached by the robber.
		Before Robber reaches this component, the largest component of the graph contained even more vertices.
		If the cops move back to a position $X$ with only one good component, the escape space of the robber is again larger than $\sizeLargeComponent$.
		If the cops move outside of the robber components, its size does not change.
		The only way to shrink the size of the escape space is to move the next cop into the robber component.
		Such a move could, in addition to the vertex where the cop is placed onto, remove some connected component from the robber escape space.
		But we know from \cref{lem:third-comp-single} that this component can only be a singleton as long as the larger remaining component is good.
		Thus from this point onward the size of the robber component shrinks by at most two per round as long as the robber is still in a good component.
		Since 
		\begin{equation*}
		\sizeLargeComponent - 2(q' - h)> \sizeLargeComponent -2\left(\frac{\ell h - h^2 + 2h}{4} - h\right) = \sizeNonGood,
		\end{equation*}
		Robber can choose to remain in a good component by \cref{lem:size-non-good}.
	\end{claimproof}
	As $\sizeNonGood>1$, for $h>1$, Robber is not caught after $q\leq \frac{\ell h - h^2 + 2h}{4}$ rounds and he wins the game.
\end{proof}

The bound of \cref{lem:lower-bound-rounds} turns out to be tight up to an additive term that only depends on $h$, for all $h>3$.

\begin{lemma}
	\label{lem:upper-bound-rounds}
	For $3<h<\ell-3$ and $q\geq \frac{\ell h}{4} + h +1$, Cops wins the game $\CR_q^{h+1}(\grid{h}{\ell})$.
\end{lemma}

\begin{proof}
	In the first $h$ rounds, Cops places the cops on $X\coloneqq\{(1,\lfloor\frac{\ell}{2}\rfloor-\lfloor\frac{h}{2}\rfloor+1),(2,\lfloor\frac{\ell}{2}\rfloor-\lfloor\frac{h}{2}\rfloor+2),\ldots,(h,\lfloor\frac{\ell}{2}\rfloor+\lceil\frac{h}{2}\rceil)\}$.
	
	We assume that the robber is in position $\gamma_{(h,1)}^X$.
	In the next round Cops places the last cop on $(h,\lfloor\frac{\ell}{2}\rfloor+\lceil\frac{h}{2}\rceil - 2)$.
	In round $r$, as long as $(1,1)\notin X_r$, there is some $i\in[h]$ and $1<j\leq \lfloor\frac{\ell}{2}\rfloor-\lfloor\frac{h}{2}\rfloor+1$ such that the cops are in position $X_{r}=\{(1,j), (2,j+1), \ldots, (i,j+i-1), (i,j+i-3), (i+1,j+i-2), \ldots, (h,j+h-3)\}$.
	The robber is either still in position $\gamma_{(h,1)}^{X_{r}}$ or he moved to the singleton component $\{(i,j+i-2)\}$.
	If the robber moves to the singleton component, he can be caught in the next move as $h+1>4$ and the singleton has at most four neighbours.
	Thus assume that the robber is in position $\gamma_{(h,1)}^{X_{r}}$.
	Then the cop player moves the cop on position $(i,j+i-1)$ to $(i-1,\max(1,j+i-4))$ if $i>1$, or to position $(h,j+h-5)$.
	That is in $h$ rounds the cops move iteratively to a new diagonal that is two columns closer to vertex $(h,1)$.
	Else if $(1,1)\in X_r$, Cops continues this strategy of iteratively moving the (now partial) diagonal two columns closer to the vertex $(h,1)$ until the robber is caught.
	All in all this strategy places $\lceil(\lfloor\frac{\ell}{2}\rfloor-\lfloor\frac{h}{2}\rfloor+i)/2\rceil$ times a cop in row $i\in [h]$, if the robber avoids singleton components as long as possible.
	Therefore the cops can win in
	\begin{align*}
	1 + \sum_{i=1}^{h} \lceil(\lfloor\frac{\ell}{2}\rfloor-\lfloor\frac{h}{2}\rfloor+i)/2\rceil &\leq 1 + \frac12 \sum_{i=1}^{h} (\frac{\ell}{2}-\frac{h-1}{2}+i+1)\\
		&= 1 + \frac{\ell h-h^2+3h}{4} +\frac12\sum_{i=1}^{h} i = \frac{\ell h}{4} + h + 1
	\end{align*}
	
	If the robber starts in position $\gamma_{(1,\ell)}^X$ the strategy for the cops is symmetric, but since $\lfloor\frac{\ell}{2}\rfloor-\lfloor\frac{h}{2}\rfloor+1\geq \ell - \lfloor\frac{\ell}{2}\rfloor - \lceil\frac{h}{2}\rceil$, Cops does not need longer to catch the robber.
\end{proof}

Using this we can construct a graph $G\in\TW_{k-1}\cap\TD_q$ such that Robber wins $\CR_q^k(G)$, if $q$ is sufficiently larger than $k$.

\begin{lemma}
	\label{lem:CR_tw-td}
	For $q\geq 3$ and $2\leq k-1\leq\frac{q}{3 +\log q}$, there exist
	\begin{itemize}
		\item a connected graph $G\in\TW_{k-1}\cap\TD_q$ such that Robber wins $\CR_q^k(G)$ and
		\item a connected graph $G'\in \TW_{1}\cap\TD_q$ such that Robber wins $\CR_q^2(G')$.
	\end{itemize}
\end{lemma}

\begin{proof}
	We first construct $G'\in \TW_{1}\cap\TD_q$.
	Consider the path $G'\coloneqq\grid{1}{2^q-1}$.
	It is well known that $\tw(\grid{1}{2^q-1})=1$ and $\td(\grid{1}{2^q-1})\leq \lceil \log(2^q-1+1)\rceil = q$, thus $\grid{1}{2^q-1}\in \TW_{1}\cap\TD_q$.
	By \cref{lem:rounds-path}, Robber wins against two cops if the game is played for $q \leq \lceil \frac{2^q-1-1}{2}\rceil=2^{q-1} -1$ rounds.
	Thus Robber wins $\CR_q^2(G')$ since $q \geq 3$ by assumption.
	
	Next we construct $G\in\TW_{k-1}\cap\TD_q$.
	Since $\frac{q}{k-1} \geq \log q > 1$, we may pick an integer~$\ell$ such that $\frac{q}{k-1}(k+1) \leq \ell \leq \frac{q}{k-1}(k+2)$.  
	Define the graph $G$ as the grid $\grid{k-1}{\ell}$.
	Then $\td G \leq (k-1)(\lceil \log (\ell+1)\rceil)$, cf.\@ \cref{fig:treedepthexample}.
	Using that $\ell \leq \frac{q}{k-1}(k+2)$ and $q \geq k+1$, we get that 
	\begin{align*}
	(k-1)(\lceil \log (\ell+1)\rceil) &\leq (k-1)\left( \log \left(\frac{q}{k-1}(k+2) + 1 \right)+1\right)\\
	&\leq (k-1)\left( \log (q)+\log\left(\frac{k+2}{k-1} + \frac{1}{k+1} \right)+1\right)\\
	&\leq (k-1)\left( \log (q)+2.5 \right) \leq q,
	\end{align*}
	where the penultimate inequality holds since $k \geq 3$ by assumption.
	Hence, $G \in \TD_q$ and clearly $G \in \TW_{k-1}$.
	
	Finally, we derive that the bound in \cref{lem:lower-bound-rounds} applies, yielding that Robber wins the game $\CR_q^k(G)$.
	Since $\ell \geq \frac{q}{k-1}(k+1)$ and $q \geq k-1$.
	\begin{align*}
		\numberOfRoundsK 
		& \geq \frac{q(k+1) - (k-1)(k-3)}{4} \\ 
		& = q + \frac{(k-1)(k-3) - (k-1)(k-3)}{4} \\
		& \geq q.
	\end{align*}
	Hence, Robber wins $\CR_q^k(G)$.
\end{proof}

Observe that by attaching a complete binary tree of depth $q$ to the centre vertex of the grid, one may force the graph $G$ in \cref{lem:CR_tw-td} to have treedepth exactly $q$.

\ekqtwtdsyntax*

\begin{proof}
	We use $G$ and $G'$ from \cref{lem:CR_tw-td} and observe that if Robber wins $\CR_q^k({H})$, for any $q,k\geq 1$ and any graph $H$, he also wins $\monCR_q^k(H)$.
	Thus from \cref{lem:Ekq-cops} it follows that $G'\notin\EParam{2}{q}$ and $G\notin\Ekq$.
\end{proof} 
\section{Proofs of \Cref{thm:guardedEkq_vs_guarded-logic,thm:ckq_equivalence}}
\label{app:homind}
\subsection{Proofs of \Cref{subsec:ckq}}
\label{app:c-homind}

\homcountsInCkq*

\begin{proof}
    Since $F$ is \elimOrd{k}{q}, there is a $k$-construction tree $(T, \lambda, r)$ for $F$ of \elimDepth{} at most $q$. We construct $\phi_m$ per induction over the structure of $T$. Without loss of generality, we assume that $T$ is a binary tree.

    Let $v \in V(T)$ be a leaf of $T$. Since $H \coloneqq \lambda(v)$ is fully labelled, for any graph $G$ with $L_{H} \subseteq L_G$ there is either a unique homomorphism from $H$ to
$G$ or none at all. We thus let $\phi_1^v$ be the conjunction of formulae $x_i = x_j$ for $\nu_{H}(i) = \nu_{H}(j)$ and $Ex_ix_j$ for $\nu_{H}(i)\nu_{H}(j) \in E(H)$. We then let $\phi^v_0 = \neg \phi^v_1$ and $\phi^v_m = \bot$ for $m > 1$. Note that for all $m$, $\phi^v_m$ uses at most $k$ distinct variables and has quantifier-rank $0$, so $\phi_m^v \in \lC^k_q$.

    Now let $v \in V(T)$ be a node with two children $w_1, w_2$. Since $T$ is a construction tree, we have $\lambda(v) = \lambda(w_1)\lambda(w_2)$. Per induction hypothesis, there are formulae $\phi^1_m, \phi^2_m \in \lC^k_q$, such that $\hom(\lambda(w_i), G) = m$ if and only if $H \models \phi^i_m$ for appropriately labelled graphs $G$ and $i \in \{1,2\}$. If $m \geq 1$ we let $\phi_m$ be the disjunction of formulae $\phi^1_{m_1} \land \phi^2_{m_2}$ for all $m_1, m_2$ with $m = m_1m_2$. For $m = 0$ we let $\phi_m = \phi^1_0 \lor \phi^2_0$. This boolean combination does not alter the quantifier rank, so $\phi_m$ still has quantifier-rank at most $q$. Moreover, $\phi^1_m$ and and $\phi^2_m$ both have variables among $x_1, \dots, x_k$, so $\phi_m \in \lC^k_q$.

    Finally, suppose $v \in V(T)$ has only one child $w$. This implies that we can obtain $\lambda(v)$ from $\lambda(w)$ by removing a label $\ell$. Per induction hypothesis, there are formulae $\phi'_m \in \lC^k_q$ that capture that there are $m$ homomorphisms from $F'$ to an appropriately labelled graph $H$. By \cref{prp:homlabeldel}, we can define $\phi_m$ as the disjunction over all decompositions $m = \sum_{i=1}^t c_im_i$, for $c_i, m_i \in \N$ and $c \coloneqq \sum c_i$, of formulae
    \begin{equation*}
      \exists^{= c} x_\ell ~\neg\phi'_0 \land \bigwedge_{i \in [t]} \exists^{=c_i}x_l ~\phi'_{m_i}.
\end{equation*}

    If $v$ is the root of $T$, i.e. $\lambda(v) = F$, we let $\phi_m = \phi_m^v$.
    Observe that each elimination step increases the quantifier-rank of $\phi_m$ by one, and at the leafs the quantifier-rank is 0. Since $T$ has \elimDepth{} at most $q$, it holds that $\phi_m \in \lC^k_q$.
  \end{proof}

  \qgFromCkq*

  \begin{proof}
    The proof is per induction over the structure of $\phi$. If $\phi = [x_i = x_j]$, we let $\qg{F} = F$ be the graph consisting of a single vertex $v$ with $\nu_F(i) = \nu_F(j) = v$. If $\phi = Ex_ix_j$, we let $\qg{F} = F$ be the graph consisting of two adjacent vertices $v_1, v_2$ with $\nu_F(i) = v_1$ and $\nu_F(j) = v_2$, unless $i = j$, in which case we let $\qg{G} = 0$. It is not hard to see that there exists a (unique) homomorphism from $\qg{F}$ to a loopless graph $G$ iff $G \models \varphi$. In all these cases $\qr(\phi) = 0$, and since $G$ is always fully labelled, it holds that $F$ is \elimOrd{k}{0}. Consequently, $\qg{F} \in \R\LParam{k}{0}$.

    If $\phi = \neg \psi$, then there exists per induction hypothesis an $\qg{F}_\psi \in \R\Lkq$ modelling $\psi$ for graphs of order $n$. We use the interpolation construction from \cref{lem:qginterpolation} and let $\qg{F} = \qgp{F_\psi}{0}{1}$. Since $\Lkq$ is closed under taking products, we have $\qg{F} \in \R\Lkq$.

    If $\phi = \psi \lor \theta$, let $\qg{F_\psi}, \qg{F_\theta}$ be defined as above. Then $\qg{F} \coloneqq \qgp{F'}{1, 2}{0}$ where $\qg{F'} = \qg{F_\psi} + \qg{F_\theta}$ models $\phi$ for graphs of order $n$. Again, by \cref{lem:qginterpolation} it is $\qg{F} \in \R\Lkq$.

    Finally, consider the case $\phi = \exists^{\geq t} x_\ell \psi$. Let $\qg{F_\psi} = \sum_{i}c_iF_{\psi, i}$ be a linear combination modelling $\psi$ for graphs of size $n$. Since $\qr(\psi) = \qr(\phi) - 1$, we may assume that $\qg{F_\psi} \in \R\LParam{k}{q-1}$.
    We let $\qg{F_\psi'}$ be the graph obtained from $\qg{F_\psi}$ by removing the label $\ell$ from all $F_{\psi, i}$. Then
        \begin{align*}
          \hom(\qg{F_\psi'}, G) &= \sum_{v \in V(G)}\sum_i c_i\hom(F_{\psi, i}, G(\ell \to v))\\
          &= \sum_{v \in V(G)}\hom(\qg{F_\psi}, G(\ell \to v)),
        \end{align*}
    so $\qg{F} = \qgp{F_\psi'}{t, \dots, n}{0, \dots, t-1}$ models $\phi$ for graphs of order $n$. Moreover, it is easy to see that the $F_{\psi, i}'$, obtained by removing a label from $F_{\psi, i}$, are \elimOrd{k}{q}. Consequently, $\qg{F} \in \R\Lkq$.
  \end{proof}

  \subsection{Proofs of \Cref{subsec:gckq}}
  \label{app:gc-homind}

  \gcHomcap*

    \begin{proof}
    The construction proceeds along the same lines as \cref{lem:homcounts_in_ckq}. In fact, we only have to reconsider label deletions. Suppose $F$ is obtained from a graph $F'$ by removing a label $\ell$. Then per induction hypothesis there exist formulae $\varphi'_m \in \mathsf{GC}_{q-1}^k$ encoding the number of homomorphisms from $F'$ to an arbitrary graph $G$. Moreover, per definition there is a label $\ell'$ in $F'$ with $\nu(\ell)\nu(\ell') \in E(F')$. We then define the formula $\varphi_m \in \lGC_{q}^k$ as the disjunction over all decompositions $m = \sum_i c_im_i$, $c = \sum_i c_i$ of formulae
    \begin{equation*}
      \theta \coloneqq \exists^{=c}x_\ell (Ex_{\ell}x_{\ell'} \land \neg \varphi'_0) \land \bigwedge_i \exists^{=c_i}x_\ell(Ex_{\ell}x_{\ell'} \land \varphi'_{m_i}).
    \end{equation*}
    Note that this is only differs from the construction in \cref{lem:homcounts_in_ckq} by the added guards $Ex_{\ell}x_{\ell'}$. To see that this does not limit the strength of the formula, suppose
    \begin{equation*}
      G \models \theta' \coloneqq \exists^{=c}x_\ell \neg \varphi'_0 \land \bigwedge_i \exists^{=c_i}x_\ell\varphi'_{m_i}.
    \end{equation*}
    Then per definition there exist exactly $c_i$ vertices $v_1, \dots, v_{c_i}$ such that for each $j \in [c_i]$ it holds that $G(\ell \to v_j) \models \varphi'_{m_i}$. Per induction hypothesis this is equivalent to $\hom(F', G(\ell \to v_j)) = m_i > 0$. But for each $h \in \HOM(F', G(\ell \to v_j))$, we have $h(\nu(\ell))h(\nu(\ell')) = v_jh(\nu(\ell')) \in E(G(\ell \to v_j))$. Consequently $G(\ell \to v_j) \models Ex_{\ell}x_{\ell'}$ for all $j$, and thus $G \models \theta$.

    Conversely, if $G \models \theta$ but $G \not\models \theta'$ then there must exist a vertex $w$ such that $G(\ell \to w) \models \varphi'_{m_i}$ but $G(\ell \to w) \not\models Ex_{\ell}x_{\ell'}$. But then again per induction hypothesis there exists a homomorphism from $F'$ to $G(\ell \to w)$, and thus $w\nu(\ell') \in E(G(\ell \to w))$ and $G(\ell \to w) \models Ex_{\ell}x_{\ell'}$.
  \end{proof}

  \qgFromGckq*

  \begin{proof}
    Consider the construction used in the proof of \cref{lem:qg_from_ckq}. We show that the same construction yields an $\qg{F} \in \R\GEkq$ when restricting ourselves to $\lGC$ formulae. First note that $\GEkq$ is still closed under products, so we can use the interpolation construction without any restrictions. It thus suffices to consider the case that $\varphi = \exists^{\geq t} x_\ell (Ex_{\ell}x_{\ell'} \land \theta)$.

    Per assumption there then exists an $\qg{F_\theta} \in \R\GEParam{k}{q-1}$ such that for graphs $G$, $\hom(\qg{F_\theta}, G) = 1$ iff $G \models \theta$. Moreover, there is a graph that models $Ex_{\ell}x_{\ell'}$ -- namely the two vertex graph $F_E = (\{x, y\}, {xy})$ with $\nu^{F_E}(\ell) = x$ and $\nu^{F_E}(\ell') = y$.

    The product $\qg{F'_\theta} = \qg{F}_\theta \cdot F_E$ then models $Ex_{\ell}x_{\ell'} \land \theta$. But note that in all graphs of $\qg{F'_\theta}$, $\nu(\ell)$ has a labelled neighbour -- namely $\nu(\ell')$. This means that the linear combination $\qg{F''}$ obtained by removing the label $\ell$ from all graphs in the linear combination is still in $\R\GEkq$. Completely analogously to \cref{lem:qg_from_ckq},
    \begin{equation*}
      \qg{F_\varphi} \coloneqq \qgp{F_\theta''}{t, \dots, n}{0, \dots, t-1},
    \end{equation*}
    models $\varphi$ for graphs of size $n$.
  \end{proof}

  \subsection{Removing labels to prove \Cref{thm:guardedEkq_vs_guarded-logic}}

  We prove the result that removing the labels from $\GEkq$ yields the desired equivalence relation---and thus \Cref{thm:guardedEkq_vs_guarded-logic}---in two steps.

  \begin{lemma}\label{lem:labelsFw}
  Let $G, H$ be unlabelled graphs with $G \equiv_{\lGC^k_q} H$. Then $\hom(\GEkqLL, G) = \hom(\GEkqLL, H)$.
\end{lemma}

  \begin{proof}
  Let $G, H$ be unlabelled graphs with $G \equiv_{\lGC^k_q} H$. Then there exists a bijection $f \colon V(G) \to V(H)$ such that for all $\psi(x) \in \lGC^k_q$ it holds that
  \begin{equation*}
    G, v \models \psi(x) \iff H, f(v) \models \psi(x)
  \end{equation*}
  for all $v \in V(G)$. This holds in particular for the formulae $\phi^F_m(x_\ell)$ encoding homomorphism counts from a graph $F \in \GEkq$. We thus get
  \begin{align*}
    \hom(F, G(\ell \to v)) = m &\iff G, v \models \phi^F_m(x_\ell)\\
                               &\iff H, f(v) \models \phi^F_m(x_\ell)\\
                               &\iff \hom(f, H(\ell \to f(v))) = m.
  \end{align*}
  or equivalently
  \begin{equation*}
    \hom(F, G(\ell \to v)) = \hom(F, H(\ell \to f(v))).
  \end{equation*}

  Now for some $F \in \GEkq$, we let $F^-$ be the graph obtained from $F$ by removing all labels. Then $F^- \in \GEkqLL$, and moreover for each $X' \in \GEkqLL$ there is a graph $X \in \GEkq$ with $X' = X^-$. For an arbitrary $F^- \in \GEkqLL$ it then holds
  \begin{align*}
    \hom(F^-, G) &= \sum_{v \in V(G)}\hom(F, G(\ell \to v))\\
                 &= \sum_{v \in V(G)} \hom(F, H(\ell \to f(v)))\\
                 &= \sum_{w \in V(H)} \hom(F, H(\ell \to w))\\
                 &= \hom(F^-, H).
  \end{align*}
\end{proof}

\begin{lemma}[Folklore] \label{lem:interpolation}
	Let $I$ and $J$ be finite sets.
	Let $\mathcal{F}$ be a set of pairs of functions $(a, b)$ where $a \colon I \to \mathbb{R}$ and $b \colon J \to \mathbb{R}$.
	Suppose that $\mathcal{F}$ is closed under multiplication, i.e.\@ if $(a, b), (a', b') \in \mathcal{F}$ then $(a \cdot a', b \cdot b') \in \mathcal{F}$ where $a \cdot a'$ denotes the point-wise product of $a$ and $a'$.
	Then the following are equivalent:
	\begin{enumerate}
		\item For all $(a, b) \in \mathcal{F}$, $\sum_{i \in I} a(i) = \sum_{j \in J} b(j)$,
		\item There exists a bijection $\pi \colon I \to J$ such that $a = b \circ \pi$ for all $(a,b) \in \mathcal{F}$.
	\end{enumerate}
\end{lemma}
\begin{proof}
	The backward implication is immediate.
	Conversely, consider the vector spaces $A \leq \mathbb{R}^I$ and $B \leq \mathbb{R}^J$ spanned by the $a$ and $b$ from $\mathcal{F}$, respectively. By a Gram--Schmidt argument \cite[Lemma~29]{grohe_homomorphism_2021_arxiv}, there exists a unitary map $U \colon A \to B$ such that $Ua = b$ for all $(a,b) \in \mathcal{F}$.
	Define an equivalence relation on $I$ by $i  \sim i' \iff a(i) = a(i')$ for all $a$ and analogously on $J$. A colour class indicator vector is a vector in $\mathbb{R}^I$ (or $\mathbb{R}^J$) which is one on one of the equivalence classes and zero everywhere else.
	By arguments from \cite[Lemma~35]{grohe_homomorphism_2021_arxiv}, $U$ must send colour class indicator vectors to colour class indicator vectors. Hence, $\pi$ can be extracted from $U$.
\end{proof}

\begin{lemma}\label{lem:labelsBw}
  Let $G, H$ be unlabelled graphs with $\hom(\GEkqLL, G) = \hom(\GEkqLL, H)$. Then $G \equiv_{\lGC^k_q} H$.
\end{lemma}
\begin{proof}
	Write $\mathcal{F} $ for the set of pairs of functions $V(G) \to \mathbb{R}$, $v \mapsto \hom(F, G(1 \to v))$ and $V(H) \to \mathbb{R}$, $v \mapsto \hom(F, H(1 \to v))$ for all $F \in \GEkq$ with one label. Then $\mathcal{F}$ is as in \cref{lem:interpolation}. Hence, there exists a bijection $\pi \colon V(G) \to V(H)$ such that $\hom(F, G(1 \to v)) = \hom(F, H(1 \to \pi(v)))$ for all $v \in V(G)$ and $F \in \GEkqLL$. Hence, \cref{cor:gc-thm-with-labels} yields the claim.
\end{proof}

\Cref{thm:guardedEkq_vs_guarded-logic} is an immediate corollary of \Cref{lem:labelsFw} and \ref{lem:labelsBw}

\section{Proofs of \Cref{thm:td-closed,thm:Ekq_tw-td-semantics}}
\label{app:td-closed}
We start by recalling a construction of highly similar graphs from \cite{roberson_oddomorphisms_2022} which is reminiscent of the CFI-construction \cite{cai_optimal_1992}.
Let $G$ be a graph with $U \subseteq V(G)$. Write $\delta_{v, U} \coloneqq |\{v\} \cap U|$ for every $v\in V(G)$.
The graph $G_U$ has vertices $(v, S)$ for every $v \in V(G)$ and $S \subseteq E(v)$ with $|S| \equiv \delta_{v, U} \mod 2$ where $E(v)$ denotes the set of edges incident to $v$. It contains an edge $(v, S)(u, T)$ whenever $uv \in E(G)$ and $uv \not\in S \triangle T$ where $S\triangle T$ denotes the symmetric difference of $S$ and $T$. Write $\rho \colon G_u \to G$ for the homomorphism sending $(v, S)$ to $v$.

Note that the same construction also appears in \cite{Furer01} and may be referred to as CFI-construction with inner vertices only. Central is the following result of \cite{roberson_oddomorphisms_2022} regarding homomorphisms to these graphs.

\begin{lemma}[{\cite[Corollary~3.7 and Theorem~3.13]{roberson_oddomorphisms_2022}}] \label{lem:roberson3.7}
	For a connected graph $G$ and $U \subseteq V(G)$, the following are equivalent:
	\begin{enumerate}
		\item $|U|$ is even,
		\item $G_\emptyset \cong G_U$,
		\item $\hom(G, G_\emptyset) = \hom(G, G_U)$.
	\end{enumerate}
	Furthermore, $\hom(F, G_\emptyset) \geq \hom(F, G_U)$ for every graph $F$.
\end{lemma}
In virtue of \cite[Lemma~3.2]{roberson_oddomorphisms_2022}, we may write $G_0$ for $G_\emptyset$ and $G_1$ for any of the isomorphic $G_U$ with $U \subseteq V(G)$ of odd size.
Given \cref{lem:roberson3.7}, it suffices to prove the following proposition:

\begin{proposition}\label{prop:closedness-td-core}
	Let $k \geq 1$ and $q \geq 0$.
	Let $G$ be a connected graph.
	If Robber wins the non-monotone $\CR^k_q(G)$ then $G_0 \equiv_{\Ekq} G_1$.
\end{proposition}

Before we turn to the proof of \cref{prop:closedness-td-core}, we show how it implies \cref{thm:td-closed}, cf.\@ \cite[Theorem~4.7]{roberson_oddomorphisms_2022}. For two graphs $F_1$ and $F_2$, write $F_1 + F_2$ for their disjoint union.

\begin{proposition} \label{prop:closed-core}
	Let $\mathcal{F}$ be a graph class satisfying the following:
	\begin{enumerate}
		\item $\mathcal{F}$ is closed under disjoint unions, i.e.\@ if $F_1, F_2 \in \mathcal{F}$ then $F_1 + F_2 \in \mathcal{F}$,
		\item $\mathcal{F}$ is closed under taking summands, i.e.\@ if $F_1 + F_2 \in \mathcal{F}$ then $F_1, F_2 \in \mathcal{F}$,
		\item for every connected $G \not\in \mathcal{F}$, it holds that $G_0 \equiv_{\mathcal{F}} G_1$.
	\end{enumerate}
	Then $\mathcal{F}$ is homomorphism distinguishing closed.
\end{proposition}
\begin{proof}
	 It has to be shown that for every $F \not\in \mathcal{F}$ there exist graphs $G$ and $H$ such that
	 $G \equiv_{\mathcal{F}} H$ and $\hom(F, G) \neq \hom(F, H)$.
	 
	 Let $F \not\in \mathcal{F}$ be arbitrary. Write $F = F^1 + \dots + F^r$ as disjoint union of its connected components. Since $\mathcal{F}$ is closed under taking disjoint unions, we may suppose wlog that $F^1 \not\in \mathcal{F}$.
	 By the third assumption, $F^1_0 \equiv_{\mathcal{F}} F^1_1$.
	 Let $n \geq 1$ be large enough such that $\hom(F, K_n) > 0$.
	 By the second assumption and \cite[Theorem~5]{seppelt_logical_2023}, also $F^1_0 + K_n \equiv_{\mathcal{F}} F^1_1 + K_n$. Furthermore, applying well-known facts \cite[(5.28)--(5.30)]{lovasz_large_2012},
	 \begin{align*}
	 	\hom(F, F^1_0 + K_n) 
	 	&= \prod_{i = 1}^r \hom(F^i, F^1_0 + K_n)
	 	= \prod_{i = 1}^r \left(\hom(F^i, F^1_0) + \hom(F^i, K_n)\right) \\
	 	& > \prod_{i = 1}^r \left(\hom(F^i, F^1_1) + \hom(F^i, K_n)\right) = \hom(F, F^1_1 + K_n)
	 \end{align*}
 	since $\hom(F^i, F^1_0) \geq \hom(F^i, F^1_1)$ for every $i \in [r]$ and $\hom(F^1, F^1_0) > \hom(F^1, F^1_1)$ by \cref{lem:roberson3.7}.
 	Hence, $G \coloneqq F^1_0 + K_n$ and $H \coloneqq F^1_1 + K_n$ are as desired.
\end{proof}

\begin{proof}[Proof of \cref{thm:td-closed} assuming \cref{prop:closedness-td-core}]
	The class $\mathcal{TD}_q$ is closed under disjoint unions and taking summands.
	Furthermore, if $G \not\in \mathcal{TD}_q$ then Robber wins the non-monotone $\CR^q_q(G)$ by \cref{lem:games}. Hence, $G_0 \equiv_{\mathcal{T}_q^q} G_1$ and $\mathcal{TD}_q = \mathcal{T}_q^q$.
	Now \cref{prop:closed-core} implies the result.
\end{proof}

\begin{proof}[Proof of \cref{thm:Ekq_tw-td-semantics} assuming \cref{prop:closedness-td-core}]
	For the inclusion $\cl(\Ekq) \subseteq \mathcal{TW}_{k-1} \cap \mathcal{TD}_q$, first observe that \cref{cor:Ekq-subset-tw-td} implies that $\cl(\Ekq) \subseteq \cl(\mathcal{TW}_{k-1} \cap \mathcal{TD}_q)$.
	By \cite{neuen_homomorphism-distinguishing_2023} and \cref{thm:td-closed}, $\mathcal{TW}_{k-1}$ and $\mathcal{TD}_q$ are homomorphism distinguishing closed. By \cite[Lemma~6.1]{roberson_oddomorphisms_2022}, so is their intersection. Hence, $\cl(\Ekq) \subseteq \cl(\mathcal{TW}_{k-1} \cap \mathcal{TD}_q) = \mathcal{TW}_{k-1} \cap \mathcal{TD}_q$.
	
	By \cref{lem:CR_tw-td}, there exists a connected graph $G \in \mathcal{TW}_{k-1} \cap \mathcal{TD}_q$ on which robber wins the non-monotone $\CR^k_q(G)$. By \cref{prop:closedness-td-core}, the graphs $G_0$ and $G_1$ are homomorphism indistinguishable over $\Ekq$. By \cref{lem:roberson3.7}, $\hom(G, G_0) \neq \hom(G, G_1)$. Hence, $G \not\in \cl(\Ekq)$.
\end{proof}

\subsection{How Robber Guides Duplicator (Proof of \cref{prop:closedness-td-core})}

Let $G$ and $H$ be graphs and $k \geq 1$.
For a partial function $\gamma \colon [k] \rightharpoonup V(G) \times V(H)$, write $\gamma_G \colon [k] \rightharpoonup V(G)$ and $\gamma_H \colon [k] \rightharpoonup V(H)$ for the maps obtained from $\gamma$ by projecting to the respective components. Then $\gamma$ is a \emph{partial isomorphism} if for all $i, j \in \dom(\gamma)$, $\gamma_G(i) = \gamma_G(j) \Leftrightarrow \gamma_H(i) = \gamma_H(j)$ and $\gamma_G(i)\gamma_G(j) \in E(G) \Leftrightarrow \gamma_H(i)\gamma_H(j) \in E(H)$. Note that $\gamma = \emptyset$ is a partial isomorphism for any two graphs $G$ and $H$.

The \emph{bijective $k$-pebble game} on graphs $G$ an $H$ is played by two players Spoiler and Duplicator. The positions are partial functions $\gamma \colon [k] \rightharpoonup V(G) \times V(H)$.
If $\gamma$ is not a partial isomorphism then Spoiler wins in $0$ rounds.
If $|V(G)| \neq |V(H)|$ then Spoiler wins in $1$ round.
At any round $i \geq 1$ starting in position $\gamma$,
\begin{itemize}
	\item Spoiler picks $p \in [k]$,
	\item Duplicator supplies a bijection $f \colon V(G) \to V(H)$,
	\item Spoiler picks $v \in V(G)$.
\end{itemize}
The position is updated to the partial map $\gamma' \coloneqq \gamma[p \mapsto (v, f(v))]$ whose domain is $\dom(\gamma) \cup \{p\}$ and
\[
	\gamma' \colon q \mapsto \begin{cases}
		(v, f(v)), & \text{if } q = p,\\
		\gamma(q), & \text{otherwise}.
	\end{cases}
\]
Spoiler wins after round $i$ if $\gamma'$ is not a partial isomorphism. Otherwise Duplicator wins the $i$-round bijective $k$-pebble game.

The following \cref{thm:hella} is implicit in many sources \cite{cai_optimal_1992,hella_logical_1996}. As we were not able to find an exact reference, we choose to give a proof sketch:

\begin{theorem} \label{thm:hella}
	Let $k\geq 1$ and $q \geq 0$. Let $G$ and $H$ be graphs and $\gamma \colon [k] \rightharpoonup V(G) \times V(H)$. Then the following are equivalent:
	\begin{enumerate}
		\item Duplicator wins the $q$-round bijective $k$-pebble game starting in position $\gamma$,
		\item For all formulae $\phi(\boldsymbol{x}) \in \mathsf{C}^k_q$ with $|\gamma|$ free variables, $G, \gamma_G \models \phi(\boldsymbol{x})$ if and only if $H, \gamma_H \models \phi(\boldsymbol{x})$.
	\end{enumerate}
\end{theorem}
\begin{proof}
	By induction on $q$. For $q = 0$, Duplicator wins the $0$-round game iff the starting position is a partial isomorphism. This is precisely what can be defined using formulae in $\mathsf{C}^{k}_0$ with $|\gamma|$ free variables.
	
	Let $q > 1$. Suppose Duplicator wins the $q$-round bijective game from $\gamma$. Let $\phi(\boldsymbol{x}) \in \mathsf{C}^k_q$ be a formula with $|\gamma|$ free variables. We may suppose without loss of generality that $\phi$ is of the form $\exists^{=n} x_p.\ \psi(\boldsymbol{x})$ for some $n \geq 0$ and $p \in [k]$. If Spoiler picks $p$ then Duplicator can supply a bijection $f \colon V(G) \to V(H)$ such that for all $v \in V(G)$, Duplicator wins the $(q-1)$-round bijective game from $\gamma[p \mapsto (v, f(v))]$. By the inductive hypothesis, $G, \gamma_G[p \mapsto v] \models \psi(\boldsymbol{x})$ if and only if $H, \gamma_H[p \mapsto f(v)] \models \psi(\boldsymbol{x})$ for all $v \in V(G)$ and all $\psi \in \mathsf{C}^k_{q-1}$. Hence, $G, \gamma_G \models \phi(\boldsymbol{x})$ if and only if $H,\gamma_H \models \phi(\boldsymbol{x})$ as desired.
	
	Conversely, suppose that $G, \gamma_G \models \phi(\boldsymbol{x})$ if and only if $H, \gamma_H \models \phi(\boldsymbol{x})$ for all $\phi \in \mathsf{C}^k_q$.
	We claim that Duplicator wins the $q$-round game starting in $\gamma$.
	Suppose Spoiler picks $p \in [k]$.
Then there exists a bijection $f \colon V(G) \to V(H)$ such that $G, \gamma_G[p \mapsto v] \models \psi(\boldsymbol{x})$ iff $H, \gamma_H[p \mapsto f(v)] \models \psi(\boldsymbol{x})$ for all $\psi \in \mathsf{C}^k_{q-1}$ and all $v \in V(G)$.
Duplicator may play this bijection. For any choice $v \in V(G)$ of Spoiler, it follows inductively that Duplicator wins the $(q-1)$-round game from $\gamma[p \mapsto (v, f(v))]$.
\end{proof}

\begin{lemma}[{\cite[Lemma~4.3]{neuen_homomorphism-distinguishing_2023}, cf.\@ \cite{dawar_power_2007}}]
	\label{lem:neuen4.3}
	Let $G$ be a connected graph. 
	Let $u, v\in V(G)$. Let $P$ be a path in $G$ from $u$ to $v$.
	Then there exists an isomorphism $\phi \colon G_{\{ u \}} \to G_{\{v\}}$ such that 
	\begin{enumerate}
		\item $\rho(\phi(w, S)) = w$ for all $(w, S) \in V(G_{\{u\}})$ and
		\item $\phi(w, S) = (w, S)$ for all $(w, S) \in V(G_{\{u\}})$ with $w \in V(G) \setminus P$.
	\end{enumerate}
\end{lemma}

We can now conduct the proof of \cref{prop:closedness-td-core} following \cite[Lemma~4.4]{neuen_homomorphism-distinguishing_2023}.

\begin{proof}[Proof of \cref{prop:closedness-td-core}]
	Given the winning strategy for Robber in the non-monotone $\CR^k_q(G)$, we provide a winning strategy for Duplicator in the $q$-round bijective $k$-pebble game on $G_\emptyset$ and $G_U = G_{\{u_0\}}$ for some fixed vertex $u_0 \in V(G)$ with initial position $\emptyset$. Clearly, $|V(G_\emptyset)| = |V(G_U)|$, so the game commences.

	For a position $\gamma \colon [k] \rightharpoonup V(G_\emptyset) \times V(G_{\{u_0\}})$, we regard the partial map $\rho \circ \gamma_{G_\emptyset}$ as a positions of the cops in $\CR^k_q(G)$. More precisely, the cops are placed on $\{\rho(\gamma_{G_\emptyset}(p)) \mid p \in \dom(\gamma)\}$ and on any selection of $k - |\dom(\gamma)|$ many  vertices from the auxiliary clique.

	Throughout the game, Duplicator maintains the following invariant. Before the $i$-th round of the bijective $k$-pebble game, $1 \leq i \leq q$, with current position $\gamma \colon [k] \rightharpoonup V(G) \times V(H)$:
	There is a vertex $u \in V(G)$ and an isomorphism $\phi \colon G_{\{u\}} \to G_{\{u_0\}}$ such that
	\begin{enumerate}[I]
		\item $\rho(\phi(w, S)) = w$ for all $(w, S) \in V(G_{\{u\}})$,\label{inv1}
		\item $u$ does not appear in the image of $\rho \circ \gamma_{G_\emptyset} = \rho \circ \gamma_{G_{\{u_0\}}}$,\label{inv2}
		\item $\phi$ sends pebbled vertices to pebbled vertices, i.e.\@ $\phi \circ \gamma_{G_\emptyset} = \gamma_{G_{\{u_0\}}}$,\label{inv3}
		\item Robber wins the non-monotone $\CR^k_{q-i+1}(G)$ with cops placed on $\rho \circ \gamma_{G_\emptyset}$ and robber starting in $u \in V(G)$.\label{inv4}
	\end{enumerate}

	To clarify the indexing, observe that initially (before the first round) with position $\gamma = \emptyset$ we require Robber to be able to win $\CR^k_q(G)$ when all cops are placed on the auxiliary clique.

	Initially, let $u \in V(G)$ denote the vertex chosen by Robber when the cops are places on the auxiliary clique. 
	Let $P$ denote a shortest path from $u$ to $u_0$.
	By \cref{lem:neuen4.3}, there exists an isomorphism $\phi \colon G_{\{u\}} \to G_{\{u_0\}}$ satisfying all stipulated properties.
	
	The invariant is maintained as follows: Let $\gamma$ denote the current position and write $u \in V(G)$ for the current position of the robber. When Spoiler picks a pebble $p \in [k]$, Duplicator constructs a bijection $f \colon V(G_\emptyset) \to V(G_{\{u_0\}})$ as follows:
	Let $v \in V(G)$ be arbitrary and write $u' \in V(G)$ for the vertex where robber would escape to when Cops updates his position to $(\rho \circ \gamma_{G_\emptyset})[p \mapsto v]$. Let $P_v$ denote a shortest path from $u$ to $u'$ avoiding the vertices which appear in $\img (\rho \circ \gamma_{G_\emptyset})[p \mapsto v] \cap \img \rho \circ \gamma_{G_\emptyset}$.
	Let $\psi_v$ denote the isomorphism $G_{\{u'\}} \to G_{\{u\}}$ from \cref{lem:neuen4.3} for $P_v$.
	Duplicator plays $f \colon x \mapsto \phi(\psi_{\rho(x)}(x))$.
	
	Now Spoiler picks a vertex $x \in V(G_\emptyset)$. Then the map $\phi' \coloneqq \phi \circ \psi_{\rho(x)}$ satisfies the properties of the invariant and $f(x) = \phi'(x)$. Indeed, \cref{inv1} is immediate from \cref{lem:neuen4.3}. \cref{inv2} holds because the robber was not yet captured, i.e.\@ $\rho(x) \neq u'$ in the notation from above. \cref{inv3} holds since $P_v$ avoids all vertices in $\img (\rho \circ \gamma_{G_\emptyset})[p \mapsto v] \cap \img \rho \circ \gamma_{G_\emptyset}$ and $f(x) = \phi'(x)$. Clearly, \cref{inv4} holds.
	
	It remains to argue that the updated $\gamma$ is a partial isomorphism. By the properties of the invariant, $\phi \colon G_{\{u\}} \to G_{\{u_0\}}$ is an isomorphism and $u$ does not appear in the image of $\rho \circ \gamma_{G_\emptyset} = \rho \circ \gamma_{G_{\{u_0\}}}$. Hence, $\phi$ restricts to an isomorphism $G_\emptyset - \rho^{-1}(u) \to G_{\{u_0\}} - \rho^{-1}(u)$. The partial map $\gamma$ coincides with this isomorphism by \cref{inv3}.
\end{proof}
 

\newpage

\begin{thebibliography}{10}

\bibitem{abramsky_pebbling_2017}
Samson Abramsky, Anuj Dawar, and Pengming Wang.
\newblock The {Pebbling} {Comonad} in {Finite} {Model} {Theory}.
\newblock In {\em Proceedings of the 32nd {Annual} {ACM}/{IEEE} {Symposium} on
  {Logic} in {Computer} {Science}}, {LICS} '17. IEEE Press, 2017.
\newblock event-place: Reykjavík, Iceland.
\newblock \href {https://doi.org/10.1109/LICS.2017.8005129}
  {\path{doi:10.1109/LICS.2017.8005129}}.

\bibitem{AbramskyDW17}
Samson Abramsky, Anuj Dawar, and Pengming Wang.
\newblock The pebbling comonad in finite model theory.
\newblock In {\em 32nd Annual {ACM/IEEE} Symposium on Logic in Computer
  Science, {LICS} 2017, Reykjavik, Iceland, June 20-23, 2017}, pages 1--12.
  {IEEE} Computer Society, 2017.
\newblock \href {https://doi.org/10.1109/LICS.2017.8005129}
  {\path{doi:10.1109/LICS.2017.8005129}}.

\bibitem{abramsky_relating_2021}
Samson Abramsky and Nihil Shah.
\newblock Relating structure and power: {Comonadic} semantics for computational
  resources.
\newblock {\em Journal of Logic and Computation}, 31(6):1390--1428, September
  2021.
\newblock \href {https://doi.org/10.1093/logcom/exab048}
  {\path{doi:10.1093/logcom/exab048}}.

\bibitem{bodlaender_partial_1998}
Hans~L. Bodlaender.
\newblock A partial $k$-arboretum of graphs with bounded treewidth.
\newblock {\em Theoretical Computer Science}, 209(1):1--45, December 1998.
\newblock URL:
  \url{https://www.sciencedirect.com/science/article/pii/S0304397597002284},
  \href {https://doi.org/10.1016/S0304-3975(97)00228-4}
  {\path{doi:10.1016/S0304-3975(97)00228-4}}.

\bibitem{cai_optimal_1992}
Jin-Yi Cai, Martin Fürer, and Neil Immerman.
\newblock An optimal lower bound on the number of variables for graph
  identification.
\newblock {\em Combinatorica}, 12(4):389--410, December 1992.
\newblock \href {https://doi.org/10.1007/BF01305232}
  {\path{doi:10.1007/BF01305232}}.

\bibitem{dawar_lovasz-type_2021}
Anuj Dawar, Tomá\v{s} Jakl, and Luca Reggio.
\newblock Lovász-{Type} {Theorems} and {Game} {Comonads}.
\newblock In {\em 2021 36th {Annual} {ACM}/{IEEE} {Symposium} on {Logic} in
  {Computer} {Science} ({LICS})}, pages 1--13, June 2021.
\newblock \href {https://doi.org/10.1109/LICS52264.2021.9470609}
  {\path{doi:10.1109/LICS52264.2021.9470609}}.

\bibitem{dawar_power_2007}
Anuj Dawar and David Richerby.
\newblock The power of counting logics on restricted classes of finite
  structures.
\newblock In Jacques Duparc and Thomas~A. Henzinger, editors, {\em Computer
  Science Logic}, pages 84--98. Springer Berlin Heidelberg, 2007.
\newblock \href {https://doi.org/10.1007/978-3-540-74915-8_10}
  {\path{doi:10.1007/978-3-540-74915-8_10}}.

\bibitem{dell_lovasz_2018}
Holger Dell, Martin Grohe, and Gaurav Rattan.
\newblock Lovász {Meets} {Weisfeiler} and {Leman}.
\newblock {\em 45th International Colloquium on Automata, Languages, and
  Programming (ICALP 2018)}, pages 40:1--40:14, 2018.
\newblock URL: \url{http://drops.dagstuhl.de/opus/volltexte/2018/9044/}, \href
  {https://doi.org/10.4230/LIPICS.ICALP.2018.40}
  {\path{doi:10.4230/LIPICS.ICALP.2018.40}}.

\bibitem{dvorak_recognizing_2010}
Zden{\v e}k Dvo{\v r}{\'a}k.
\newblock On recognizing graphs by numbers of homomorphisms.
\newblock {\em Journal of Graph Theory}, 64(4):330--342, August 2010.
\newblock \href {https://doi.org/10.1002/jgt.20461}
  {\path{doi:10.1002/jgt.20461}}.

\bibitem{Furer01}
Martin F{\"{u}}rer.
\newblock {Weisfeiler-Lehman Refinement Requires at Least a Linear Number of
  Iterations}.
\newblock In Fernando Orejas, Paul~G. Spirakis, and Jan van Leeuwen, editors,
  {\em Automata, Languages and Programming, 28th International Colloquium,
  {ICALP} 2001, Crete, Greece, July 8-12, 2001, Proceedings}, volume 2076 of
  {\em Lecture Notes in Computer Science}, pages 322--333. Springer, 2001.
\newblock \href {https://doi.org/10.1007/3-540-48224-5\_27}
  {\path{doi:10.1007/3-540-48224-5\_27}}.

\bibitem{giannopoulou_lifo-search_2012}
Archontia~C. Giannopoulou, Paul Hunter, and Dimitrios~M. Thilikos.
\newblock {LIFO}-search: {A} min{\textendash}max theorem and a searching game
  for cycle-rank and tree-depth.
\newblock {\em Discrete Applied Mathematics}, 160(15):2089--2097, October 2012.
\newblock \href {https://doi.org/10.1016/j.dam.2012.03.015}
  {\path{doi:10.1016/j.dam.2012.03.015}}.

\bibitem{Gottlob03}
Georg Gottlob, Nicola Leone, and Francesco Scarcello.
\newblock Robbers, marshals, and guards: game theoretic and logical
  characterizations of hypertree width.
\newblock {\em Journal of Computer and System Sciences}, 66(4):775--808, 2003.
\newblock Special Issue on PODS 2001.
\newblock URL:
  \url{https://www.sciencedirect.com/science/article/pii/S0022000003000308},
  \href {https://doi.org/10.1016/S0022-0000(03)00030-8}
  {\path{doi:10.1016/S0022-0000(03)00030-8}}.

\bibitem{grohe_counting_2020}
Martin Grohe.
\newblock Counting {Bounded} {Tree} {Depth} {Homomorphisms}.
\newblock In {\em Proceedings of the 35th {Annual} {ACM}/{IEEE} {Symposium} on
  {Logic} in {Computer} {Science}}, {LICS} '20, pages 507--520, New York, NY,
  USA, 2020. Association for Computing Machinery.
\newblock event-place: Saarbr{\"u}cken, Germany.
\newblock \href {https://doi.org/10.1145/3373718.3394739}
  {\path{doi:10.1145/3373718.3394739}}.

\bibitem{grohe_word2vec_2020}
Martin Grohe.
\newblock word2vec, node2vec, graph2vec, {X2vec}: {Towards} a {Theory} of
  {Vector} {Embeddings} of {Structured} {Data}.
\newblock In Dan Suciu, Yufei Tao, and Zhewei Wei, editors, {\em Proceedings of
  the 39th {ACM} {SIGMOD}-{SIGACT}-{SIGAI} {Symposium} on {Principles} of
  {Database} {Systems}, {PODS} 2020, {Portland}, {OR}, {USA}, {June} 14-19,
  2020}, pages 1--16. ACM, 2020.
\newblock \href {https://doi.org/10.1145/3375395.3387641}
  {\path{doi:10.1145/3375395.3387641}}.

\bibitem{grohe_logic_2021}
Martin Grohe.
\newblock The {Logic} of {Graph} {Neural} {Networks}.
\newblock In {\em 36th {Annual} {ACM}/{IEEE} {Symposium} on {Logic} in
  {Computer} {Science}, {LICS} 2021, {Rome}, {Italy}, {June} 29 - {July} 2,
  2021}, pages 1--17. IEEE, 2021.
\newblock \href {https://doi.org/10.1109/LICS52264.2021.9470677}
  {\path{doi:10.1109/LICS52264.2021.9470677}}.

\bibitem{grohe_homomorphism_2021_arxiv}
Martin Grohe, Gaurav Rattan, and Tim Seppelt.
\newblock Homomorphism {Tensors} and {Linear} {Equations}.
\newblock November 2021.
\newblock URL: \url{http://arxiv.org/abs/2111.11313}.

\bibitem{grohe_homomorphism_2022}
Martin Grohe, Gaurav Rattan, and Tim Seppelt.
\newblock {Homomorphism Tensors and Linear Equations}.
\newblock In Miko{\l}aj Boja\'{n}czyk, Emanuela Merelli, and David~P. Woodruff,
  editors, {\em 49th International Colloquium on Automata, Languages, and
  Programming (ICALP 2022)}, volume 229 of {\em Leibniz International
  Proceedings in Informatics (LIPIcs)}, pages 70:1--70:20, Dagstuhl, Germany,
  2022. Schloss Dagstuhl -- Leibniz-Zentrum f{\"u}r Informatik.
\newblock \href {https://doi.org/10.4230/LIPIcs.ICALP.2022.70}
  {\path{doi:10.4230/LIPIcs.ICALP.2022.70}}.

\bibitem{hella_logical_1996}
Lauri Hella.
\newblock Logical {Hierarchies} in {PTIME}.
\newblock {\em Information and Computation}, 129(1):1--19, August 1996.
\newblock \href {https://doi.org/10.1006/inco.1996.0070}
  {\path{doi:10.1006/inco.1996.0070}}.

\bibitem{KreutzerO11}
Stephan Kreutzer and Sebastian Ordyniak.
\newblock Digraph decompositions and monotonicity in digraph searching.
\newblock {\em Theor. Comput. Sci.}, 412(35):4688--4703, 2011.
\newblock \href {https://doi.org/10.1016/j.tcs.2011.05.003}
  {\path{doi:10.1016/j.tcs.2011.05.003}}.

\bibitem{lovasz_operations_1967}
L{\'a}szl{\'o} Lov{\'a}sz.
\newblock Operations with structures.
\newblock {\em Acta Mathematica Academiae Scientiarum Hungarica},
  18(3):321--328, September 1967.
\newblock \href {https://doi.org/10.1007/BF02280291}
  {\path{doi:10.1007/BF02280291}}.

\bibitem{lovasz_large_2012}
László Lovász.
\newblock {\em Large networks and graph limits}.
\newblock Number volume 60 in American {Mathematical} {Society} colloquium
  publications. American Mathematical Society, Providence, Rhode Island, 2012.
\newblock \href {https://doi.org/10.1090/coll/060}
  {\path{doi:10.1090/coll/060}}.

\bibitem{mancinska_quantum_2020}
Laura Man{\v c}inska and David~E. Roberson.
\newblock Quantum isomorphism is equivalent to equality of homomorphism counts
  from planar graphs.
\newblock In {\em 2020 {IEEE} 61st {Annual} {Symposium} on {Foundations} of
  {Computer} {Science} ({FOCS})}, pages 661--672, 2020.
\newblock \href {https://doi.org/10.1109/FOCS46700.2020.00067}
  {\path{doi:10.1109/FOCS46700.2020.00067}}.

\bibitem{neuen_homomorphism-distinguishing_2023}
Daniel Neuen.
\newblock Homomorphism-{Distinguishing} {Closedness} for {Graphs} of {Bounded}
  {Tree}-{Width}, April 2023.
\newblock arXiv:2304.07011 [cs, math].
\newblock URL: \url{http://arxiv.org/abs/2304.07011}.

\bibitem{nguyen_graph_2020}
Hoang Nguyen and Takanori Maehara.
\newblock Graph homomorphism convolution.
\newblock In Hal~Daumé III and Aarti Singh, editors, {\em Proceedings of the
  37th International Conference on Machine Learning}, volume 119 of {\em
  Proceedings of Machine Learning Research}, pages 7306--7316. PMLR, 13--18 Jul
  2020.
\newblock URL: \url{https://proceedings.mlr.press/v119/nguyen20c.html}.

\bibitem{Rabinovich14}
Roman Rabinovich.
\newblock {\em Graph complexity measures and monotonicity}.
\newblock PhD thesis, {RWTH} Aachen University, 2013.
\newblock URL: \url{https://publications.rwth-aachen.de/record/230227}.

\bibitem{rattan_weisfeiler-leman_2023}
Gaurav Rattan and Tim Seppelt.
\newblock Weisfeiler--{Leman} and {Graph} {Spectra}.
\newblock In {\em Proceedings of the 2023 {Annual} {ACM}-{SIAM} {Symposium} on
  {Discrete} {Algorithms} ({SODA})}, pages 2268--2285. 2023.
\newblock \href {https://doi.org/10.1137/1.9781611977554.ch87}
  {\path{doi:10.1137/1.9781611977554.ch87}}.

\bibitem{roberson_oddomorphisms_2022}
David~E. Roberson.
\newblock Oddomorphisms and homomorphism indistinguishability over graphs of
  bounded degree, June 2022.
\newblock Number: arXiv:2206.10321 arXiv:2206.10321 [math].
\newblock URL: \url{http://arxiv.org/abs/2206.10321}.

\bibitem{roberson_lasserre_2023}
David~E. Roberson and Tim Seppelt.
\newblock {Lasserre Hierarchy for Graph Isomorphism and Homomorphism
  Indistinguishability}.
\newblock In Kousha Etessami, Uriel Feige, and Gabriele Puppis, editors, {\em
  50th International Colloquium on Automata, Languages, and Programming (ICALP
  2023)}, volume 261 of {\em Leibniz International Proceedings in Informatics
  (LIPIcs)}, pages 101:1--101:18, Dagstuhl, Germany, 2023. Schloss Dagstuhl --
  Leibniz-Zentrum f{\"u}r Informatik.
\newblock URL: \url{https://drops.dagstuhl.de/opus/volltexte/2023/18153}, \href
  {https://doi.org/10.4230/LIPIcs.ICALP.2023.101}
  {\path{doi:10.4230/LIPIcs.ICALP.2023.101}}.

\bibitem{seppelt_logical_2023}
Tim Seppelt.
\newblock Logical {Equivalences}, {Homomorphism} {Indistinguishability}, and
  {Forbidden} {Minors}, February 2023.
\newblock arXiv:2302.11290 [cs, math].
\newblock URL: \url{http://arxiv.org/abs/2302.11290}.

\bibitem{seymour_graph_1993}
Paul~D. Seymour and Robin Thomas.
\newblock Graph {Searching} and a {Min}-{Max} {Theorem} for {Tree}-{Width}.
\newblock {\em J. Comb. Theory, Ser. B}, 58(1):22--33, 1993.
\newblock \href {https://doi.org/10.1006/jctb.1993.1027}
  {\path{doi:10.1006/jctb.1993.1027}}.

\end{thebibliography}
\end{document}